\documentclass[12pt]{article}
\usepackage[utf8]{inputenc}

\usepackage{blindtext}
\usepackage[authoryear, round]{natbib}   
\usepackage{amsmath, amsthm, amssymb}
\usepackage{array}
\usepackage{algorithm}
\usepackage[utf8]{inputenc}
 \usepackage{quotes}
\usepackage{graphicx}
\usepackage[noend]{algpseudocode}
\usepackage[title]{appendix}

\usepackage{mathtools}
\usepackage{authblk}
\usepackage[margin=1in]{geometry}
\usepackage{lipsum}
\newcolumntype{P}[1]{>{\centering\arraybackslash}p{#1}}
\usepackage{blindtext}
\usepackage{amsfonts}
\usepackage[title]{appendix}
\usepackage{float}
\usepackage{caption}
\usepackage{subcaption}
\allowdisplaybreaks
\newtheorem{prop}{Proposition}
\newtheorem{cor}{Corollary}
\usepackage{color}

\newcommand\blfootnote[1]{%
  \begingroup
  \renewcommand\thefootnote{}\footnote{#1}%
  \addtocounter{footnote}{-1}%
  \endgroup
}
\usepackage{titling}
\date{} 

\providecommand{\keywords}[1]
{
  \small	
  \textbf{\textit{Keywords---}} #1
}
\author{
Xiang Cui and  Alexandra Chronopoulou\\ \smallskip 
University of Illinois at  Urbana-Champaign.}
\title{Optimal Sampling for Estimation of Fractional Brownian Motion}

\begin{document} 
\maketitle

\begin{abstract}
In this paper, we focus on multiple sampling problems for the estimation of the fractional Brownian motion when the maximum number of samples is limited, extending existing results in the literature in a non-Markovian framework. Two classes of sampling schemes are proposed: a deterministic scheme and a level-triggered scheme. For the deterministic sampling scheme, the sampling times are selected beforehand and do not depend on the process trajectory. For the level-triggered sampling scheme, the sampling times are the times when the process crosses predetermined thresholds. The sampling times are selected sequentially in time and depend on the process trajectory. For each of the schemes, we derive the optimal sampling times by minimizing the aggregate squared error distortion. We then show that the optimal sampling strategies heavily depend on the dependence structure of the process.
\end{abstract} \hspace{12pt}


\blfootnote{Xiang Cui is Ph.D. Student (E-mail: xiangc5@illinois.edu), and
Alexandra Chronopoulou is Clinical Associate Professor (E-mail: achronop@illinois.edu), Department
of Statistics, University of Illinois at Urbana-Champaign, Champaign, IL 61820.}
\section{Introduction}
In signal processing, the main goal is to model and analyze data representations of physical events. One of the fundamental questions in this area is how to estimate and predict the signal given its behavior in the past. From a statistical point of view, the question translates to sampling the process of interest, either offline or in real-time, and applying statistical techniques to predict its future values. However, sampling is not a trivial task and it is highly dependent on the underlying model, whether continuous or discrete, the sampling frequency, and the sampling mechanism. Furthermore, all these attributes are often bound by restrictions or limitations due to the physical problem or setup, such as the cost for sampling.

In this paper, we focus on problems where the underlying process is continuous, and there are restrictions on the number of samples that we are allowed to obtain. Of course, this problem is not new. For the networked control systems (NCSs) \citep{murray2003future, hespanha2007survey}, sensors are used to transmit samples to the supervisors, which can provide decision and control. NCSs have been widely used in various areas such as mobile sensor networks, manufacturing systems, and remote surgery. Since the sensors can only transmit a limited number of samples, the design for the sampling time can affect the estimation quality of the signal. The traditional design for the sampling time is to sample the signals equidistantly or periodically in time \citep{aastrom2013computer}, which is conventionally referred to as Riemann sampling or deterministic sampling. \citet{Kushner1964} has treated a linear optimal control problem in a finite horizon using a fixed number of deterministic sampling times. \citet{Cambanis1983} studied the best deterministic and random sampling methods for detecting  a signal in noise. However, in their approach, the sampling methods are not dependent on the signal trajectory observed by the sensor. There are several alternatives to deterministic sampling, such as sampling the system when the system has changed by a specific amount. This type of sampling is called Lebesgue sampling, event-triggered sampling, or level-triggered sampling. \citet{astrom2002} used simple systems to compare deterministic and event-triggered sampling, showing that event-triggered sampling has better performance. \citet{imer2005optimal} considered an optimal estimation problem with limited measurements, in which the stochastic process was a scalar linear system. They showed that the optimal observer policy has a solution depending on the event-triggered sampling scheme. \citet{rabi2006multiple} developed optimal multiple sampling schemes for the Brownian motion. The event-triggered approach was also used to trigger the data transmission from a sensor to a remote observer to minimize the mean squared estimation error at the observer, subject to a constraint on transmission frequency \citep{li2010event}. 

In this paper, our goal is to characterize the performance gained from different types of sampling methods when estimating a fractional process, and specifically a fractional Brownian motion (fBm), extending the results in \citet{rabi2006multiple} in a non-Markovian framework.  A fractional Brownian motion is an extension of the standard Brownian motion in the sense that it is still a Gaussian process but with a very rich dependence structure, that is determined by the Hurst parameter $H$. For a mathematical definition of the process, we refer the readers to \citet{mandelbrot1968fractional}. In the literature, the fBm has been used to model phenomena that exhibit long-range dependence, which intuitively means that the last observation is strongly correlated to the first one. In the literature, fractional Brownian motion has been used to model phenomena in economics, traffic networks, and image processing that exhibit this long memory behavior \citep{beran1994statistics, chronopoulou2012estimation, comte1998long, hurst1951long}. 

In this work, we want to study two different sampling mechanisms, deterministic and level-triggered, when estimating the fractional Brownian motion based on a sequence of discrete samples where the maximum number of samples is limited. The quality of the estimators will be assessed using the distortion of a real-time estimator of the signal over a finite horizon. Specifically, we solve a sampling design problem for estimating the fractional Brownian motion within deterministic sampling and level-triggered sampling. In deterministic sampling, the sampling time sequence is selected beforehand, and it is independent of the fBm trajectory. We illustrate the method and show how to find the optimal sampling time sequence in both one sample case and multiple samples case. We show that the optimal sampling time sequence is dependent on the Hurst parameter of the fBm. In the level-triggered sampling, the sampling time sequence is selected sequentially in time, and is dependent on the fBm trajectory. The actual sampling times are the times when the
change of the fBm trajectory crosses predetermined thresholds. We illustrate how to model this sampling strategy and to find the optimal thresholds numerically in both one sample case and multiple samples case. Based on our results, the optimal thresholds are also dependent on the Hurst parameter of the fBm. 

The remainder of this paper is organized as follows. Section \ref{section:formulation} gives the formulation and the mathematical background of the problem. Section \ref{section:deterministic} shows the proposed deterministic sampling method for one sample case and multiple samples case. Section \ref{section:level-triggered} shows the proposed level-triggered sampling method for one sample case and multiple samples case. Section \ref{section:conclusion} concludes the paper with a discussion.

\section{Problem Formulation}
\label{section:formulation}

Consider a signal $x_t$, $t \in [0, T]$, that we want to estimate. A sensor can observe the state of the signal and transmit observations at times it chooses to a supervisor. However, it can only generate at most $N$ samples to be transmitted. The sampling times $S = \{\tau_1, \tau_2, \ldots, \tau_N \}$ are an increasing sequence of times with respect to the process $x_t$ and lie within $[0,T]$.

Given the samples transmitted by the sensor at the sampling times, the least-square estimate for the state $\hat{x}_t$ at the supervisor is given by:
\begin{align*}
\hat{x}_t =
 \begin{cases}
      \mathbb{E}[x_t|\mathcal{F}_0] & \text{if $0 \leq t < \tau_1$},\\
      \mathbb{E}[x_t|\mathcal{F}_{\tau_i}] & \text{if $\tau_i \leq t < \tau_{i+1},$ for $i = 1, \ldots, N-1$},\\
      \mathbb{E}[x_t|\mathcal{F}_{\tau_N}] & \text{if $\tau_N \leq t \leq T$},
    \end{cases} 
\end{align*}
where $\mathcal{F}_t$ is the filtration generated by the values sampled up to time $t$. The estimation quality is measured by the aggregate squared error distortion defined as:
\begin{align}
J(S) &= \mathbb{E}\left[\int_0^T(x_s-\hat{x}_s)^2ds\right]\\
&= \mathbb{E}\left[\int_0^{\tau_1}(x_s-\hat{x}_s)^2ds+\sum_{i=2}^N \int_{\tau_{i-1}}^{\tau_i}(x_s-\hat{x}_s)^2ds + \int^{T}_{\tau_{N}}(x_s-\hat{x}_s)^2ds\right] \label{eq:1}.
\end{align}
In this paper, we assume that $\{x_t = B_t^H;t\ge 0\}$ is a fractional Brownian motion and consider two sampling strategies: (i) deterministic sampling; (ii) level-triggered sampling. The deterministic sampling refers to the case in which the sampling time sequence $S$ is chosen beforehand, and it is independent of the process trajectory. The sequence is chosen to minimize the expected error distortion $J$ in Equation \eqref{eq:1}. In this case, the supervisor knows the time points that the sensor will transmit samples of the process. In the case of level-triggered sampling, the sensor transmits samples of the process when the change of the process $x_t$ crosses predetermined thresholds, which implies that the sampling times depend on the actual process trajectory. 

\subsection{Mathematical  background}

As we discussed before, the fractional Brownian motion has been used to model phenomena that exhibit long memory and was first proposed by \citet{mandelbrot1968fractional}. Specifically, a fractional Brownian motion $\{x_t = B_t^H;t\ge 0\}$ with a Hurst parameter $H \in (0,1)$ is a centered, continuous,  Gaussian process with covariance structure:
\begin{align*}
\text{Cov}(B_t^H, B_s^H) = \mathbb{E}(B_t^HB_s^H) = \frac{1}{2}(t^{2H}+s^{2H}- |t-s|^{2H}), \quad t, s \ge 0.
\end{align*}
When $H = \frac{1}{2}$, the process is a standard Brownian motion. However, fBm is not a semi-martingale when $H \ne \frac{1}{2}$. Specifically, when $H \ne \frac{1}{2}$, the increments of the fBm at disjoint intervals $\{B_n^H - B_{n-1}^H\}_{n =  1, 2, 3, \ldots}$ are correlated, and their correlation is
\begin{align*}
\rho_H(n) = \frac{1}{2}((n+1)^{2H} + (n-1)^{2H} -2n^{2H}).
\end{align*}
When $H > \frac{1}{2}$, the increments at disjoint intervals are positively correlated and the process exhibits long-range dependency, in the sense that $\sum_{n=1}^{\infty} \rho_H(n) = \infty$. When $H < \frac{1}{2}$, the increments at disjoint intervals are negatively correlated and the process exhibits short memory, in the sense that 
$\sum_{n=1}^{\infty} |\rho_H(n)| < \infty$, also called a rough behavior. For more properties of fBm, we refer the readers to \citet{mandelbrot1968fractional} and \citet{nualart2006malliavin}.

\section{Deterministic Sampling}
\label{section:deterministic}
\subsection{One sample case}

In the deterministic framework, the sampling sequence is selected beforehand and hence it is independent of the trajectory of $B_t^H$. In order to illustrate the method, we first discuss the single sample case, where we can only obtain one observation at time $t = \tau_1$. 

The least square estimate for $t<\tau_1$ is: 
\[\hat{B}_t^H = \mathbb{E}[B_t^H|\mathcal{F}_0] = 0.\]
For $t \ge \tau_1$, the least square estimate is:
\begin{align*}
\hat{B}_t^H = \mathbb{E}[B_t^H|\mathcal{F}_{\tau_1}] = \mathbb{E}[B_t^H|B_{\tau_1}^H] .
\end{align*}
To compute this conditional expectation, it suffices to recall that $B_t^H \sim N(0, t^{2H})$ with auto-correlation function
\[\rho(B_t^H, B_{\tau_1}^H) = \frac{\text{Cov}(B_t^H, B_{\tau_1}^H)}{\sqrt{t^{2H}}\sqrt{\tau_1^{2H}}}.\]
For simplicity in the notation, we will refer to $\rho(B_t^H, B_{\tau_1}^H)$ by $\rho$. Given the expression for $\rho$ and the property of the Gaussian process, we can calculate the predicted value $B_t^H$ given the sample at time $\tau_1$ by:
\begin{align*}
\hat{B}_t^H = \mathbb{E}[B_t^H|B_{\tau_1}^H] = B_{\tau_1}^H \cdot \rho \cdot \frac{\sqrt{t^{2H}}}{\sqrt{\tau_1^{2H}}} = \frac{B_{\tau_1}^H}{2\tau_1^{2H}}(t^{2H}+\tau_1^{2H}-|t-\tau_1|^{2H}).
\end{align*}
Therefore, we have the following result regarding the aggregate squared error distortion.
\begin{prop}
\label{prop:deter-one-sample}
When the underlying process is a fractional Brownian motion with Hurst parameter $H\in(0,1)$, the aggregate squared error distortion with a sample obtained at time $t=\tau_1$ is given by
\begin{align}
\label{deter:distortion1}
J(\tau_1) &= \frac{T^{2H+1}}{2H+1} - \frac{1}{4\tau_1^{2H}}\int_{\tau_1}^{T} (s^{2H}+\tau_1^{2H}-|s-\tau_1|^{2H})^2ds.
\end{align}
\end{prop}

\begin{proof}
The expression for the distortion is easily obtained by direct calculations using the properties of the fBm:
\begin{align}
J(\tau_1) & = \mathbb{E}\int_0^{T}(B^H_s)^2ds -2\mathbb{E} \int_{\tau_1}^{T}(B^H_s)(\hat{B}^H_s)ds + \mathbb{E} \int_{\tau_1}^{T}(\hat{B}^H_s)^2ds \nonumber\\
&= \mathbb{E}\int_0^{T}(B^H_s)^2ds -2\mathbb{E}[\int_{\tau_1}^{T} \mathbb{E} [B^H_s\hat{B}^H_s|   \mathcal{F}_{\tau_1}]ds]+ \mathbb{E} \int_{\tau_1}^{T}(\hat{B}^H_s)^2ds \nonumber\\
& =  \int_0^{T} \mathbb{E} (B^H_s)^2ds -  \int_{\tau_1}^{T} \mathbb{E}(\hat{B}^H_s)^2ds \nonumber\\
& = \frac{T^{2H+1}}{2H+1} - \frac{1}{4\tau_1^{2H}}\int_{\tau_1}^{T} (s^{2H}+\tau_1^{2H}-|s-\tau_1|^{2H})^2ds. \nonumber
\end{align}
\end{proof}
In this deterministic context, the optimal sampling time is defined as the $\hat{\tau}_1 \in [0, T]$  that minimizes the distortion $J(\tau_1)$ in Equation \eqref{deter:distortion1}. Since we can not get a closed form of the distortion $J(\tau_1)$ in Equation \eqref{deter:distortion1}, we need to use a numerical method (e.g., \citet{brent2013algorithms}) to find the optimal $\hat{\tau}_1$. In Table \ref{table-deter-one-sample}, we present the optimal sampling time $\hat{\tau}_1$'s in the single sample case for Hurst parameter $H = 0.1, 0.2, 0.3, \ldots, 0.9$ when $T = 20$. We also present the same results in  Figure \ref{fig1}.

\begin{table}[h]
\begin{center}
\begin{tabular}{| c | c | c | c | c | c | c | c | c |}
\hline
Hurst & $\hat{\tau}_1$ & Distortion & Hurst & $\hat{\tau}_1$ & Distortion & Hurst & $\hat{\tau}_1$ & Distortion \\
\hline
0.1 & 5.658 & 23.433 & 0.2 & 8.417 & 33.685 & 0.3 & 9.469 & 48.327 \\
\hline
0.4 & 9.889 & 69.544 & 0.5 & 10 & 100 & 0.6 & 9.903 & 142.530 \\
\hline
0.7 & 9.606 & 198.413 & 0.8 & 9.028 & 260.998 & 0.9 & 7.854 & 290.634\\
\hline
\end{tabular}
\caption{\label{table-deter-one-sample} Optimal sampling time $\hat{\tau}_1$'s and the corresponding optimal distortions in the single sample case for Hurst parameter $H = 0.1, 0.2, 0.3, \ldots, 0.9$ when $T = 20$.}
\end{center}
\end{table}

\begin{figure}[h]
\begin{center}
  \includegraphics[width=0.7\linewidth]{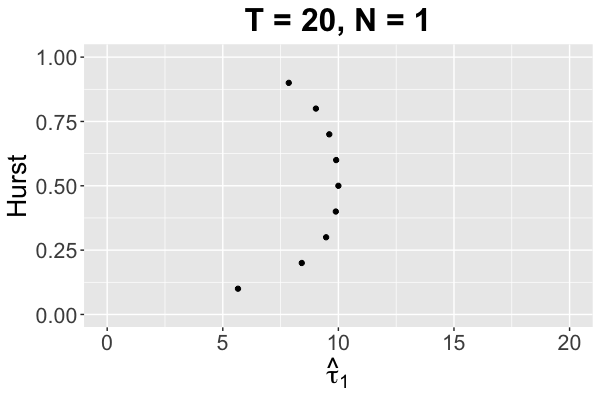}\par
  \caption{Optimal sampling time $\hat{\tau}_1$'s in the single sample case for  Hurst parameter $H = 0.1, 0.2, 0.3, \ldots, 0.9$ when $T = 20$.}
  \label{fig1}
\end{center}
\end{figure}

In Figure \ref{fig1}, the $y$-axis represents the Hurst parameter and $x$-axis represents the optimal sampling time in one sample case for the corresponding Hurst parameter. When $H = \frac{1}{2}$, the optimal sampling time $\hat{\tau}_1$ is $\frac{T}{2}$. When $H$ keeps further away from $H = \frac{1}{2}$, the optimal $\hat{\tau}_1$ becomes smaller. We recover the result of the Brownian motion in \citet{rabi2006multiple} and generalize it to fractional Brownian motion when $H \neq \frac{1}{2}$.

\subsection{Multiple samples case}

In this section, we generalize the single sample case result to $N$ samples. We compute the optimal deterministic sampling times  $0 \le \tau_1 \le \tau_2 \le \tau_3 \ldots \le \tau_N \le T$ in $[0, T]$.  For $\tau_i \le t < \tau_{i+1}$, the estimation for $\hat{B}_t^H$ is
\begin{align}
\hat{B}_t^H = \mathbb{E}(B_t^H|\mathcal{F}_{\tau_i}) = \mathbb{E}(B_t^H|B_{\tau_1}^H, B_{\tau_2}^H, \ldots, B_{\tau_i}^H),
\label{non-truncate-estimate}
\end{align}
where $\mathcal{F}_{\tau_i}$ is the filtration generated by the samples at time $\tau_1, \ldots, \tau_i$. Recall that 
\begin{align*}
(B^H_t, B^H_{\tau_1}, B^H_{\tau_2},\ldots, , B^H_{\tau_i})^T \sim N(0, \Sigma),
\end{align*}
where
\begin{align*}
\Sigma &= 
\begin{pmatrix*}
\text{Var}(B^H_t) & \text{Cov}(B^H_t, B^H_{\tau_1}) & \ldots & \text{Cov}(B^H_t, B^H_{\tau_i})\\
\text{Cov}(B^H_{\tau_1}, B^H_t) & \text{Var}(B^H_{\tau_1}) & \ldots & \text{Cov}(B^H_{\tau_1}, B^H_{\tau_i})\\
\vdots & \vdots & \ddots &\vdots \\
\text{Cov}(B^H_{\tau_i}, B^H_t) & \text{Cov}(B^H_{\tau_i}, B^H_{\tau_1}) & \ldots & \text{Var}(B^H_{\tau_i})
\end{pmatrix*} = \begin{pmatrix*}
\Sigma_{11} & \Sigma_{12}\\
\Sigma_{12} & \Sigma_{22}
\end{pmatrix*}
\end{align*}
and $\Sigma_{11} = \text{Var}(B^H_t)$.
Based on the property of the conditional distribution of the multivariate Gaussian distribution, we obtain
\begin{equation}
\hat{B}_t^H = \mathbb{E}B_t^H + \Sigma_{12}\Sigma_{22}^{-1}
(\begin{pmatrix}
B_{\tau_1}^H\\
B_{\tau_2}^H\\
\vdots \\
B_{\tau_i}^H
\end{pmatrix} - \mathbb{E}
\begin{pmatrix}
B_{\tau_1}^H\\
B_{\tau_2}^H\\
\vdots \\
B_{\tau_i}^H
\end{pmatrix}) = 
\Sigma_{12}\Sigma_{22}^{-1}
\begin{pmatrix}
B_{\tau_1}^H\\
B_{\tau_2}^H\\
\vdots \\
B_{\tau_i}^H
\end{pmatrix} \label{eq:2}
\end{equation}
and the following proposition regarding the aggregate squared error distortion follows.
\begin{prop}
When the underlying process is a fractional Brownian motion with Hurst parameter $H\in(0,1)$, the aggregate squared error distortion given $N$ samples obtained at times $\tau_1, \tau_2, \ldots, \tau_N$ is given by
\begin{align}
J(\tau_1, \tau_2, \ldots, \tau_N)  &=\frac{T^{2H+1}}{2H+1} - \sum_{i = 1}^N\int_{\tau_i}^{\tau_{i+1}} \mathbb{E}(\hat{B}_t^H)^2dt.
\label{eq:dis-non}
\end{align}
\end{prop}
\begin{proof}
\begin{align}
J(\tau_1, \tau_2, \ldots, \tau_N) &= \mathbb{E}[\sum_{i = 0}^N\int_{\tau_i}^{\tau_{i+1}}(B_t^H - \hat{B}_t^H)^2dt] \nonumber\\
&= \mathbb{E} \int_{0}^T(B_t^H)^2dt +  \mathbb{E} \sum_{i = 1}^N\int_{\tau_i}^{\tau_{i+1}} [(\hat{B}_t^H)^2-2(B_t^H\hat{B}_t^H)]dt \nonumber\\
&= \mathbb{E} \int_{0}^T(B_t^H)^2dt +  \sum_{i = 1}^N  \mathbb{E} \int_{\tau_i}^{\tau_{i+1}} [(\hat{B}_t^H)^2-2(B_t^H\hat{B}_t^H)]dt \nonumber\\
&= \mathbb{E} \int_{0}^T(B_t^H)^2dt - \sum_{i = 1}^N \mathbb{E} \int_{\tau_i}^{\tau_{i+1}} (\hat{B}_t^H)^2dt \nonumber\\
&= \frac{T^{2H+1}}{2H+1} - \sum_{i = 1}^N\int_{\tau_i}^{\tau_{i+1}} \mathbb{E}(\hat{B}_t^H)^2dt, \nonumber
\end{align}
where $\tau_0 = 0$ and $\tau_{N+1} = T$.
\end{proof}

As we can see in Equation \eqref{eq:dis-non}, the aggregate squared error distortion is still not an explicit formula and remains to be further analyzed in order to obtain the optimal sampling times. 

Before proceeding to the calculation of the optimal sampling times, we also investigate a truncated version of the problem. That is, instead of using all previous samples information like in Equation \eqref{non-truncate-estimate}, we estimate the value of $B_{t}^H$ using only the last observed value. 

In this case, the estimation $\hat{B}^H_t$ for $\tau_i < t < \tau_{i+1}$ becomes
\begin{align*}
\hat{B}_t^H =  \mathbb{E}(B_t^H|B^H_{\tau_i}) = \frac{B_{\tau_i}^H}{2\tau_i^{2H}}(t^{2H}+\tau_i^{2H}-|t-\tau_i|^{2H}).
\end{align*}
In the other words, we replaced the filtration $\mathcal{F}_{\tau_i}$ by $B_{\tau_i}^H$. We call this the truncated case for which we obtain the truncated formula for the squared error distortion.
\begin{prop}
When the underlying process is a fractional Brownian motion with Hurst parameter $H\in(0,1)$, the aggregate squared error distortion given $N$ samples obtained at times $\tau_1, \tau_2, \ldots, \tau_N$ with truncation is obtained by
\begin{align}
J(\tau_1, \tau_2, \ldots, \tau_N) &= \int_{0}^T\mathbb{E}(B_t^H)^2dt - \sum_{i = 1}^{N}\int_{\tau_i}^{\tau_{i+1}} \mathbb{E}(\hat{B}_t^H)^2dt\\
&= \frac{T^{2H+1}}{2H+1} -  \sum_{i = 1}^{N} \int_{\tau_i}^{\tau_{i+1}}\mathbb{E}(\frac{B_{\tau_i}^H}{2\tau_i^{2H}}(t^{2H}+\tau_i^{2H}-|\tau_i-t|^{2H}))^2dt, \label{two-sample-trunc-deter}
\end{align}
where 
\begin{align*}
\int_{\tau_i}^{\tau_{i+1}}\mathbb{E}(\frac{B_{\tau_i}^H}{2\tau_i^{2H}}(t^{2H}+\tau_i^{2H}-|\tau_i-t|^{2H}))^2dt = \frac{1}{4\tau_i^{2H}}\int_{\tau_i}^{\tau_{i+1}} (s^{2H}+\tau_i^{2H}-|s-\tau_i|^{2H})^2ds.
\end{align*}
\end{prop}
Investigating this version of the problem and comparing our results with the non-truncated case is important since it allows us to investigate the trade-off between computational complexity (carrying all the history of the signal) and accuracy (truncating the history of the signal).

\subsection{Computation of the optimal sampling times}

In this section, we use $N=2$ and $N=3$ as examples to illustrate how to find the optimal sets of sampling times $\tau_1, \tau_2, \ldots, \tau_N$ in $[0, T]$ for deterministic sampling in both truncated sampling case and non-truncated sampling case.

\subsubsection{Two samples case}

For $\tau_2 \le t < T$, based on Equation \eqref{eq:2}, the estimation $\hat{B}_t^H$ is
\begin{align*}
\hat{B}_t^H = C(\tau_1, \tau_2, H)(A_1(\tau_1, \tau_2, t, H)B_{\tau_1}^H - A_2(\tau_1, \tau_2, t, H)B_{\tau_2}^H),
\end{align*}
where $C(\tau_1, \tau_2, H)$, $A_1(\tau_1, \tau_2, t, H)$ and $A_2(\tau_1, \tau_2, t, H)$ are given in Section \ref{appendix:corollary1}. 

\begin{cor}
\label{cor-deter-2}
When the underlying process is a fractional Brownian motion with Hurst parameter $H\in(0,1)$, the aggregate squared error distortion for the two samples case is:
\begin{align*}
J(\tau_1, \tau_2) &=\frac{T^{2H+1}}{2H+1} - \frac{1}{4\tau_1^{2H}}\int_{\tau_1}^{\tau_2} (s^{2H}+\tau_1^{2H}-|s-\tau_1|^{2H})^2ds\\
& - C(\tau_1, \tau_2, H)^2 \int_{\tau_2}^{T}[A_1(\tau_1, \tau_2, t, H)^2\tau_1^{2H} + A_2(\tau_1, \tau_2, t, H)^2\tau_2^{2H}\\
& - A_1(\tau_1, \tau_2, t, H)A_2(\tau_1, \tau_2, t, H)(\tau_2^{2H}+\tau_1^{2H}- |\tau_2-\tau_1|^{2H})]dt.
\end{align*}
\end{cor}
The proof of Corollary \ref{cor-deter-2} is given in Section \ref{appendix:corollary1}. In order to numerically compute the optimal $\hat{\tau}_1$ and $\hat{\tau}_2$ that minimize $J(\tau_1, \tau_2)$ in Corollary \ref{cor-deter-2}, we use linear constrained optimization techniques \citep{lange1999numerical}. In the same way, we can also find the optimal $\hat{\tau}_1$ and $\hat{\tau}_2$ for the truncated optimal sampling by minimizing the distortion in Equation \eqref{two-sample-trunc-deter} for $N = 2$. In Tables \ref{table-deter-two-sample-nontrunc} and \ref{table-deter-two-sample-trunc}, we present the optimal sampling times $(\hat{\tau}_1, \hat{\tau}_2)$ in the two samples non-truncated case and truncated case for Hurst parameter $H = 0.1, 0.2, 0.3, \ldots, 0.9$ when $T = 20$. When $H = 0.5$, the two samples non-truncated case is the same as the truncated case, so we leave it empty in Table \ref{table-deter-two-sample-nontrunc}. We also illustrate the same results in Figure \ref{fig2}.

\begin{table}[h]
\begin{center}
\begin{tabular}{| c | c | c | c | c | c | }
\hline
Hurst & $(\hat{\tau}_1, \hat{\tau}_2)$ & Distortion & Hurst & $(\hat{\tau}_1, \hat{\tau}_2)$ & Distortion \\
\hline
0.1 & (3.197, 9.718) & 20.514 & 0.2 & (5.317, 12.190) & 27.851 \\
\hline
0.3 & (6.203, 12.976) & 37.408 & 0.4 & (6.569, 13.264) & 50.108  \\
\hline
0.6 & (6.582, 13.279) & 87.251 & 0.7 & (6.325, 13.122) & 110.298\\
\hline
0.8 & (5.839, 12.847) & 129.302 & 0.9 & (4.897, 12.334) & 123.020\\
\hline
\end{tabular}
\caption{\label{table-deter-two-sample-nontrunc}Optimal sampling times $(\hat{\tau}_1, \hat{\tau}_2)$ and the corresponding optimal distortions in the two samples non-truncated case for Hurst parameter $H = 0.1, 0.2, 0.3, \ldots, 0.9$ when $T = 20$.}
\end{center}
\end{table}

\begin{table}[h]
\begin{center}
\begin{tabular}{| c | c | c | c | c | c |}
\hline
Hurst & $(\hat{\tau}_1, \hat{\tau}_2)$ & Distortion & Hurst & $(\hat{\tau}_1, \hat{\tau}_2)$ & Distortion\\
\hline
0.1 & (3.172, 11.214) & 21.088 & 0.2 & (5.302, 12.597) & 28.211\\
\hline
0.3 & (6.198, 13.092) & 37.614 & 0.4 & (6.568, 13.283) & 50.176 \\
\hline
0.5 & (6.667, 13.333) & 66.667 & 0.6 & (6.581, 13.291) & 87.368\\
\hline
0.7 & (6.322, 13.161) & 110.870 & 0.8 & (5.828, 12.912) & 130.694\\
\hline
0.9 & (4.874, 12.410) & 124.947 & & &\\
\hline
\end{tabular}
\caption{\label{table-deter-two-sample-trunc}Optimal sampling times $(\hat{\tau}_1, \hat{\tau}_2)$ and the corresponding optimal distortions in the two samples truncated case for Hurst parameter $H = 0.1, 0.2, 0.3, \ldots, 0.9$ when $T = 20$.}
\end{center}
\end{table}

\begin{figure}[h]
\begin{center}
  \includegraphics[width=0.7\linewidth]{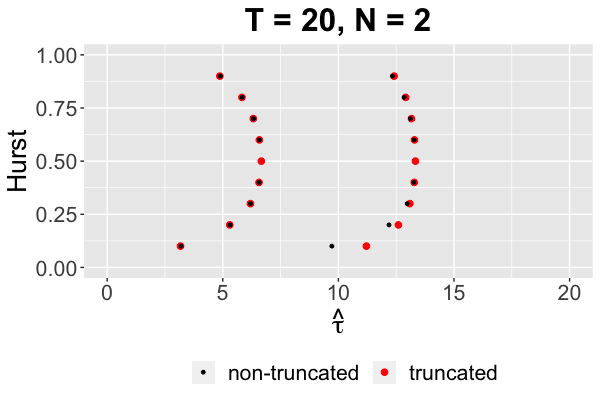}\par
  \caption{Optimal sampling times for $T=20$ in two samples cases. The $y$-axis represents the Hurst parameter and $x$-axis represents the optimal sampling times in two samples cases for the corresponding Hurst parameter. The black points represent the non-truncated case and the red points represent the truncated case.}
  \label{fig2}
\end{center}
\end{figure}

From Figure \ref{fig2}, we can see that the optimal sampling time sequence is the uniform sampling in $[0,T]$ when $H = \frac{1}{2}$. When $H$ keeps further away from $H = \frac{1}{2}$, the optimal sampling times become smaller. We observe that the optimal sampling times for the truncated case and the non-truncated case are almost the same except when $H$ is around 0.1. When $H$ is around 0.1, the optimal sampling time $\hat{\tau}_2$ in the non-truncated case is smaller than that in the truncated case.

We also compare the distortions for different Hurst parameters in Figure \ref{two-loss}. Given the same Hurst parameter, we can see the distortions for the truncated and the non-truncated sampling case are almost the same, while the non-truncated sampling case gives a lightly smaller distortion.
\begin{figure}[h]
\begin{center}
  \includegraphics[width=0.7\linewidth]{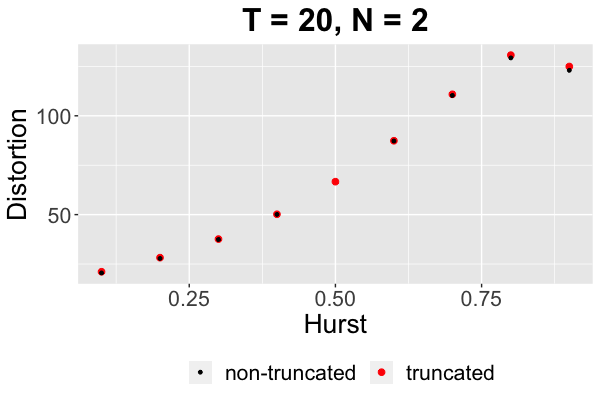}\par
  \caption{Optimal distortions for two samples cases when $T=20$. The $y$-axis represents the optimal distortion and $x$-axis represents the corresponding Hurst parameter in two samples cases. The black points represent the non-truncated case and the red points represent the truncated case.}
  \label{two-loss}
\end{center}
\end{figure}

\subsubsection{Three samples case}

For $\tau_3 \le t < T$, based on Equation \eqref{eq:2}, the estimation $\hat{B}_t^H$ is
\begin{align*}
\hat{B}_t^H &= \mathbb{E}(B_t^H|\mathcal{F}_{\tau_3}) = \mathbb{E}(B_t^H|B_{\tau_1}^H, B_{\tau_2}^H,B_{\tau_3}^H)\\
&= C(\tau_1, \tau_2, \tau_3, H)(A_1(\tau_1, \tau_2, \tau_3, t, H)B_{\tau_1}^H + A_2(\tau_1, \tau_2, \tau_3, t, H)B_{\tau_2}^H  + A_3(\tau_1, \tau_2, \tau_3, t, H)B_{\tau_3}^H).
\end{align*}
The details regarding $C(\tau_1, \tau_2, \tau_3, H)$, $A_1(\tau_1, \tau_2, \tau_3, t, H)$, $A_2(\tau_1, \tau_2, \tau_3, t, H)$ and $A_3(\tau_1, \tau_2, \tau_3, t, H)$ are given in Section \ref{appendix:corollary2}.

\begin{cor}
\label{cor-deter-3}
When the underlying process is a fractional Brownian motion with Hurst parameter $H\in(0,1)$, the aggregate squared error distortion for the three samples case is
\begin{align*}
J(\tau_1, \tau_2, \tau_3) &= \frac{T^{2H+1}}{2H+1} - \frac{1}{4\tau_1^{2H}}\int_{\tau_1}^{\tau_2} (s^{2H}+\tau_1^{2H}-|s-\tau_1|^{2H})^2ds\\
&-C(\tau_1, \tau_2, H)^2\int_{\tau_2}^{\tau_3}[A_1(\tau_1, \tau_2, t, H)^2\tau_1^{2H} + A_2(\tau_1, \tau_2, t, H)^2\tau_2^{2H}\\
& - A_1(\tau_1, \tau_2, t, H)A_2(\tau_1, \tau_2, t, H) (\tau_2^{2H}+\tau_1^{2H}- |\tau_2-\tau_1|^{2H})]dt\\
&- C(\tau_1, \tau_2, \tau_3, H)^2 \int_{\tau_3}^{T} [A_1(\tau_1, \tau_2, \tau_3, t, H)^2\tau_1^{2H}+A_2(\tau_1, \tau_2, \tau_3, t, H)^2\tau_2^{2H} \\
&+A_3(\tau_1, \tau_2, \tau_3, t, H)^{2}\tau_3^{2H}\\
&+A_1(\tau_1, \tau_2, \tau_3, t, H)A_2(\tau_1, \tau_2, \tau_3, t, H)(\tau_1^{2H}+\tau_2^{2H}-(\tau_2-\tau_1)^{2H})\\
&+A_1(\tau_1, \tau_2, \tau_3, t, H)A_3(\tau_1, \tau_2, \tau_3, t, H)(\tau_1^{2H}+\tau_3^{2H}-(\tau_3-\tau_1)^{2H})\\
&+A_2(\tau_1, \tau_2, \tau_3, t, H)A_3(\tau_1, \tau_2, \tau_3, t, H)(\tau_2^{2H}+\tau_3^{2H}-(\tau_3-\tau_2)^{2H})
]dt.
\end{align*}
\end{cor}
The proof of Corollary \ref{cor-deter-3} is given in Section \ref{appendix:corollary2}. 
Using the same techniques as those in the two samples case, we obtain the optimal sampling times  $\hat{\tau}_1$, $\hat{\tau}_2$ and $\hat{\tau}_3$ for the truncated and the non-truncated case. The results are summarized in Table \ref{table-deter-three-sample-nontrunc}, Table \ref{table-deter-three-sample-trunc} and Figure \ref{three-sample}.

\begin{table}[h]
\begin{center}
\begin{tabular}{| c | c | c | c | c | c | }
\hline
Hurst & $(\hat{\tau}_1, \hat{\tau}_2, \hat{\tau}_3)$ & Distortion & Hurst & $(\hat{\tau}_1, \hat{\tau}_2, \hat{\tau}_3)$ & Distortion \\
\hline
0.1 & (2.177,  6.559, 12.162) & 18.802 & 0.2 & (3.867,  8.849, 14.165) & 24.429 \\
\hline
0.3 & (4.607,  9.634, 14.749) & 31.249 & 0.4 & (4.916,  9.927, 14.952) & 39.730  \\
\hline
0.6 & (4.929,  9.943, 14.964) & 61.637 & 0.7 & (4.712,  9.775, 14.863) & 72.954\\
\hline
0.8 & (4.308,  9.479, 14.688) & 79.296 & 0.9 & (3.546,  8.938, 14.380) & 68.520\\
\hline
\end{tabular}
\caption{\label{table-deter-three-sample-nontrunc}Optimal sampling times $(\hat{\tau}_1, \hat{\tau}_2, \hat{\tau}_3)$ and the corresponding optimal distortions in the three samples non-truncated case for Hurst parameter $H = 0.1, 0.2, 0.3, \ldots, 0.9$ when $T = 20$.}
\end{center}
\end{table}

\begin{table}[h]
\begin{center}
\begin{tabular}{| c | c | c | c | c | c |}
\hline
Hurst & $(\hat{\tau}_1, \hat{\tau}_2, \hat{\tau}_3)$ & Distortion & Hurst & $(\hat{\tau}_1, \hat{\tau}_2, \hat{\tau}_3)$ & Distortion \\
\hline
0.1 & (2.195,  7.760, 13.839) & 19.751 & 0.2 & (3.871,  9.199, 14.605) & 25.021\\
\hline
0.3 & (4.608,  9.735, 14.872) & 31.577 & 0.4 & (4.917,  9.945, 14.974) & 39.835 \\
\hline
0.5 & (5.000, 10.000, 15.000) & 50 & 0.6 & (4.929,  9.954, 14.979) & 61.803\\
\hline
0.7 & (4.714,  9.814, 14.914) & 73.743 & 0.8 & (4.309,  9.547, 14.787) & 81.192\\
\hline
0.9 & (3.542,  9.019, 14.534) & 71.245 & & &\\
\hline
\end{tabular}
\caption{\label{table-deter-three-sample-trunc}Optimal sampling times $(\hat{\tau}_1, \hat{\tau}_2, \hat{\tau}_3)$ and the corresponding optimal distortions in the three samples truncated case for Hurst parameter $H = 0.1, 0.2, 0.3, \ldots, 0.9$ when $T = 20$.}
\end{center}
\end{table}

\begin{figure}[h]
\begin{center}
  \includegraphics[width=0.7\linewidth]{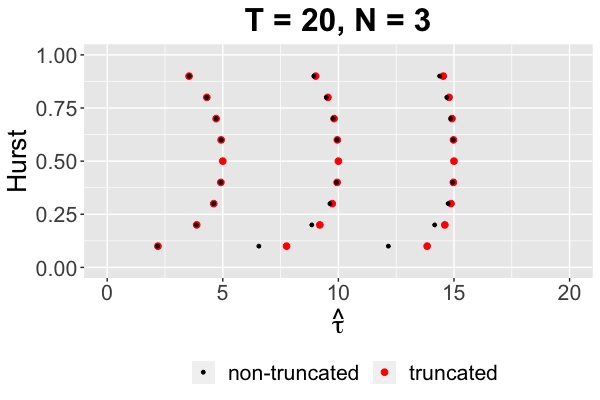}\par
  \caption{Optimal sampling times for $T=20$ in three samples cases. The $y$-axis represents the Hurst parameter and $x$-axis represents the optimal sampling times in three samples cases for the corresponding Hurst parameter. The black points represent the non-truncated case and the red points represent the truncated case.}
  \label{three-sample}
\end{center}
\end{figure}

From Figure \ref{three-sample}, we can see the optimal sampling time sequence is the uniform sampling in $[0,T]$ when $H = \frac{1}{2}$. When $H$ keeps further away from $H = \frac{1}{2}$, the optimal sampling times become smaller. We observe that the optimal sampling times for the truncated and non-truncated sampling methods are almost the same except when $H$ is around 0.1. When $H$ is around 0.1, the optimal $\hat{\tau}_2$ and $\hat{\tau}_3$ in non-truncated sampling case are both smaller than those in the truncated sampling case. We also compare the distortion for both cases in Figure \ref{three-sample-loss}. We can see that the non-truncated sampling has a smaller distortion compared with that for the truncated sampling, especially when $H$ is close to 1 or close to 0. But the distortion difference between these two is quite small.

\begin{figure}[h]
\begin{center}
  \includegraphics[width=0.7\linewidth]{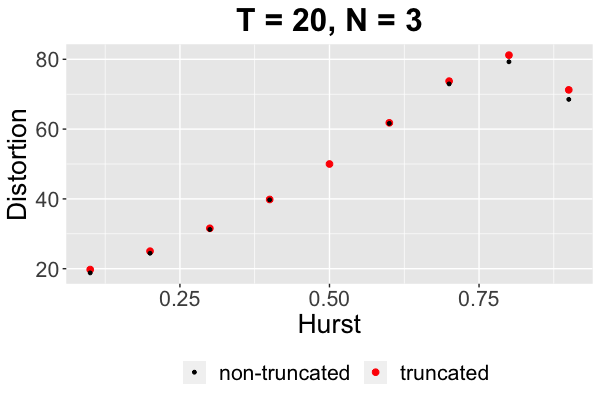}\par
  \caption{Optimal distortions for three samples cases when $T=20$. The $y$-axis represents the optimal distortion and $x$-axis represents the corresponding Hurst parameter in three samples cases. The black points represent the non-truncated case and the red points represent the truncated case.}
   \label{three-sample-loss}
\end{center}
\end{figure}

Through the results in one sample, two samples, and three samples cases, we can see that the optimal sampling times tend to be smaller when the Hurst parameter $H$ is away from 0.5. For the same number of samples allowed, the distortion has a monotone increasing trend when $H$ increases from 0.1 to 0.8 when the number of samples is bigger than 1. But it decreases when $H$ increases from 0.8 to 0.9. For the non-truncated and the truncated sampling cases, we can see that they both have similar optimal sampling times and similar optimal distortions. The non-truncated sampling has a slightly better distortion compared with the truncated sampling. As a result, the truncated sampling can be used to save some computation resources and can achieve similar performance as the non-truncated sampling.

\section{Level-Triggered Sampling}
\label{section:level-triggered}

\subsection{One sample case}

In the level-triggered sampling framework, the sampling sequence is selected sequentially, and hence it is dependent on the trajectory of $B_t^H$. The actual sampling times are the times when the change of the trajectory $B_t^H$ crosses predetermined thresholds. In order to illustrate the method, we first discuss the single sample case, where we can only obtain one observation at the time when the trajectory of $B_t^H$ crosses a predetermined threshold $\eta$.

For a given $\eta \ge 0$,  we define the time to obtain a sample as: $\tau_{\eta} = \inf_{t \ge 0}\{t: |B^H_t|\ge \eta\}$. The actual sampling time is $\tau_1 = \tau_{\eta} \land T$. Therefore, we have the following result regarding the aggregate squared error distortion.

\begin{prop}
When the underlying process is a fractional Brownian motion with Hurst parameter $H\in(0,1)$, the aggregate squared error distortion for level-triggered sampling one sample case with threshold $\eta$ is given by
\begin{align*}
J(\eta) & = \frac{T^{2H+1}}{2H+1 }- \mathbb{E}[\frac{1}{4}\eta^2((T-\tau_{1}) + \frac{T^{4H+1}/\tau_1^{4H}-\tau_1}{4H+1} + \frac{2(T^{2H+1}-\tau_1^{2H+1})/\tau_1^{2H}}{2H+1} \\
& +\frac{(T - \tau_1)^{4H+1}/\tau_1^{4H}}{4H+1} - 2\frac{(T - \tau_1)^{2H+1}/\tau_1^{2H}}{2H+1} - 2\tau_1\int_1^{\frac{T}{\tau_1}}\mu^{2H}(\mu-1)^{2H}d\mu)] \\
&=   \frac{T^{2H+1}}{2H+1 }-\mathbb{E}[h(\eta, B_t^H))] 
\end{align*}
where $\tau_{\eta} = \inf_{t \ge 0}\{t: |B^H_t|\ge \eta\}$ and $\tau_1 = \tau_{\eta} \land T$.
\label{prop:one}
\end{prop}

\begin{proof}
The expression for the distortion is easily obtained by direct calculations using the properties of the fBm:
\begin{align}
J(\eta) &  = \frac{T^{2H+1}}{2H+1 }- \mathbb{E}\int_{\tau_1}^T(\hat{B}_s^H)^2ds \nonumber\\
&= \frac{T^{2H+1}}{2H+1 }-  \mathbb{E}\int_{\tau_1}^T(\frac{B_{\tau_1}^H}{2\tau_1^{2H}}(s^{2H}+\tau_1^{2H}-|s-\tau_1|^{2H}))^2ds \nonumber\\
 &= \frac{T^{2H+1}}{2H+1 }-\frac{1}{4} \mathbb{E}[(B_{\tau_1}^H)^2((T-\tau_{1}) + \frac{T^{4H+1}/\tau_1^{4H}-\tau_1}{4H+1} + \frac{2(T^{2H+1}-\tau_1^{2H+1})/\tau_1^{2H}}{2H+1} \nonumber\\
& +\frac{(T - \tau_1)^{4H+1}/\tau_1^{4H}}{4H+1} - 2\frac{(T - \tau_1)^{2H+1}/\tau_1^{2H}}{2H+1} - 2\tau_1\int_1^{\frac{T}{\tau_1}}\mu^{2H}(\mu-1)^{2H}d\mu)] \nonumber\\
 &= \frac{T^{2H+1}}{2H+1 }- \mathbb{E}[\frac{1}{4}\eta^2((T-\tau_{1}) + \frac{T^{4H+1}/\tau_1^{4H}-\tau_1}{4H+1} + \frac{2(T^{2H+1}-\tau_1^{2H+1})/\tau_1^{2H}}{2H+1} \nonumber\\
& +\frac{(T - \tau_1)^{4H+1}/\tau_1^{4H}}{4H+1} - 2\frac{(T - \tau_1)^{2H+1}/\tau_1^{2H}}{2H+1} - 2\tau_1\int_1^{\frac{T}{\tau_1}}\mu^{2H}(\mu-1)^{2H}d\mu)] 
 \label{eq:3} \nonumber\\
 &=   \frac{T^{2H+1}}{2H+1 }- \mathbb{E}[h(\eta, B_t^H)].
\end{align}
\end{proof}

We can write $F(\eta) =  \mathbb{E}[h(\eta, B_t^H)]$. In order to find the optimal threshold $\hat{\eta}$ which can minimize the $J(\eta)$ in Proposition \ref{prop:one}, we need to find the $\hat{\eta}$ that maximizes $F(\eta)$. But the form of $F(\eta)$ is not completely known with respect to $\eta$ and only observation with variation $h(\eta, B_t^H)$ calculated using a fBm sample path $B_t^H$ can be made. In order to find the point of optimum of $J(\eta)$, we need to use a stochastic approximation algorithm. 

\subsubsection{Stochastic optimization}

\citet{kiefer1952} proposed a stochastic approximation algorithm to find the point that maximizes a function observed with variation (noise). Using Kiefer-Wolfowitz (KW) Algorithm, we can initialize parameter $\eta$ as $\eta_0$. For each iteration, we update parameter as
\begin{align*}
\eta^{(n+1)} = \eta^{(n)} + a^{(n)}Y^{(n)}
\end{align*}
with
\begin{align*}
Y^{(n)} &= \frac{h(\eta^{(n)} + c^{(n)},  \delta_n^+) - h(\eta^{(n)}- c^{(n)}, \delta_n^-)}{2c^{(n)}},
\end{align*}
where $\delta_n^+$ and $\delta_n^-$ are two independent fBm sample paths and $h(\eta,  \delta)$ is the noise-corrupted observation using the fBm sample path $\delta$, and $\{a^{(n)}\}$ and $\{c^{(n)}\}$ are two real-valued, deterministic tuning sequences. Since using only two fBm sample paths at each iteration to calculate $Y^{(n)}$ is too noisy and does not have a good finite time performance, we use a group of multiple fBm sample paths at each iteration and take the average of them to calculate $Y^{(n)}$. This is also suggested by \citet{robbins1951stochastic}.

Several papers in the literature \citep{D1957, fabian1967stochastic, polyak1990optimal} have focused on establishing bounds on the mean-squared error (MSE) $\mathbb{E}(\eta^{(n)}-\eta^*)^2$, where $\eta^*$ is the maximum point of $F(\eta)$, and deriving the optimal rates at which the MSE converges to zero. \citet{broadie2011general} proposed an adaptive version of the KW algorithm, which can improve its finite-time performance. They proposed the optimal choice of the tuning sequences as $a^{(n)} = \alpha/(n + \beta)$ and $c^{(n)} = \gamma/n^{1/4}$ for some $\alpha, \gamma \in \mathbb{R}_+ \text{ and } \beta \in \mathbb{Z}_+$. Their method dynamically scales and shifts the $a^{(n)}$ and $c^{(n)}$ sequences by adaptively adjusting $\alpha$, $\beta$ and $\gamma$ to better match them with the characteristics and noise level of the observation $h(\eta, \delta)$, and thus improve the finite time performance. We run Broadie method \citep{broadie2011general} until the sequence $\{\eta^{(n)}\}$ converges to the optimum, which gives us the optimal threshold $\hat{\eta}$ that minimizes $J(\eta)$. 

In Table \ref{table-level-one-sample}, we present the optimal thresholds $\hat{\eta}$ in the single sample case for Hurst parameter $H = 0.1, 0.2, \ldots, 0.7$ when $T=20$. We recover the optimal $\hat{\eta}$ result for Brownian motion in \citet{rabi2006multiple} and generalize it to the fractional Brownian motion when $H \neq \frac{1}{2}$. As we can see in Table \ref{table-level-one-sample}, when $H$ increases from 0.1 to 0.7, the optimal threshold $\hat{\eta}$ keeps increasing. This shows that as the Hurst parameter $H$ increases, the optimal threshold becomes larger. As $H$ increases, fBm has less fluctuation and has a smoother behavior, which leads the fBm to go toward one direction. Thus, it is reasonable to have a larger optimal threshold $\hat{\eta}$ when $H$ increases. 
 
\begin{table}[h]
\begin{center}
\begin{tabular}{| c | c | c | c | c | c | c | c | c |}
\hline
Hurst & $\hat{\eta}$ & Distortion & Hurst & $\hat{\eta}$ & Distortion & Hurst & $\hat{\eta}$ & Distortion \\
\hline
0.1 & 2.364 & 2.395 & 0.2 & 2.666 & 9.404 & 0.3 & 3.024 & 22.116 \\
\hline
0.4 & 3.500 & 43.521 & 0.5 & 4.176 & 78.963 & 0.6 & 4.850 & 134.631\\
\hline
0.7 & 5.741 & 216.328 & & & & & &\\
\hline
\end{tabular}
\caption{\label{table-level-one-sample}Optimal thresholds $\hat{\eta}$ and corresponding minimized distortions in the single sample case for  Hurst parameter $H = 0.1, 0.2, 0.3, \ldots, 0.7$  when $T = 20$.}
\end{center}
\end{table}

\subsubsection{Optimal thresholds given H for fBm}
\label{section:relationship}

Here, we observe the relation between the optimal thresholds and the time horizon $T$ given the Hurst parameter $H$'s. For $T = 20$, we use Broadie method \citep{broadie2011general} to estimate the optimal thresholds $\hat{\eta}$ for the one sample case when the Hurst parameter $H$ ranges from 0.05 to 0.75 with a step of 0.05. 

\begin{figure}[h]
\centering
  \includegraphics[width=0.5\linewidth]{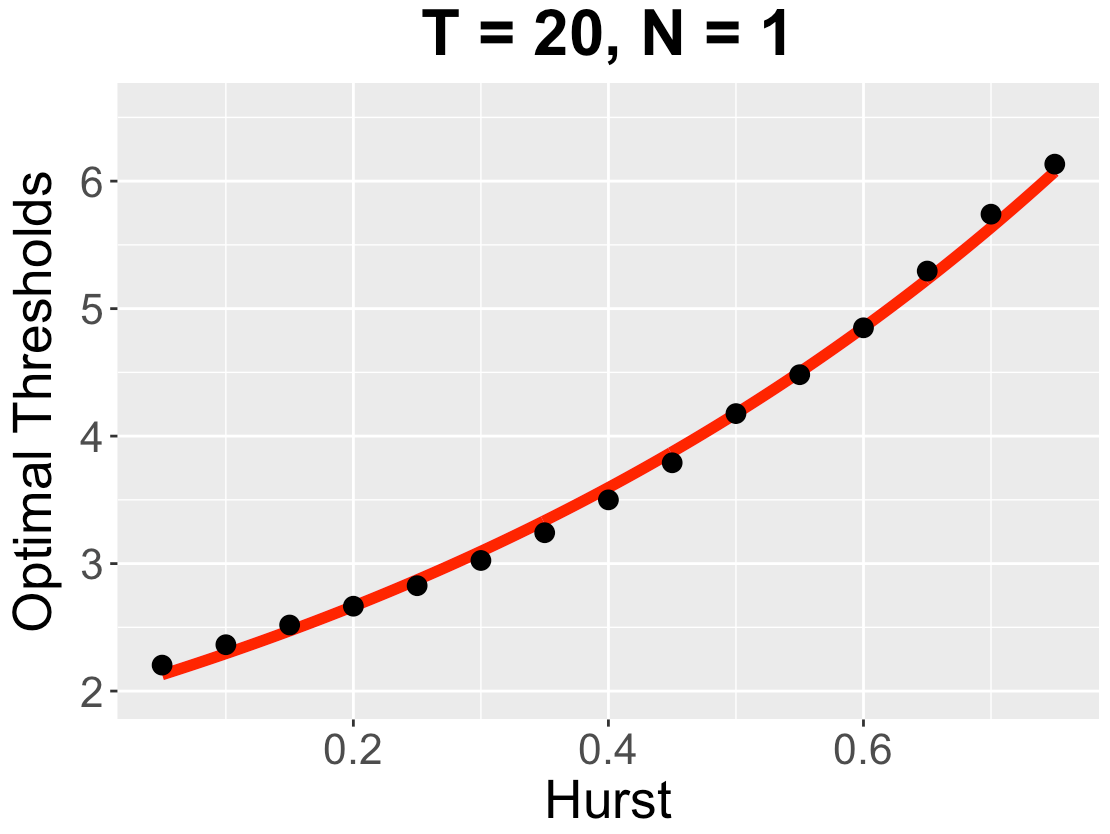}
  \caption{Optimal thresholds $\hat{\eta}$'s  for $H = 0.05, 0.1, 0.15,\ldots, 0.75$ when $T = 20$. The black points are the optimal threshold for a given Hurst parameter. The $x$-axis represents the Hurst parameter and the $y$-axis represents the optimal threshold for the given Hurst parameter. The red line is the relation between the optimal thresholds and the time horizon for different Hurst parameter $H$'s.}
  \label{fig4}
\end{figure}

We plot the optimal thresholds $\hat{\eta}$ for different Hurst parameter $H$'s in Figure \ref{fig4}. The black points in Figure \ref{fig4} are the optimal thresholds for the corresponding Hurst parameters. The red line is the relation we observed between the optimal threshold and the time horizon, which is presented in Equation \eqref{eq:5}:
\begin{align}
\hat{\eta} = C \cdot T^{\frac{1}{4} (2H + 1)},  
\label{eq:5}
\end{align}
where the constant value $C$ is the same for different Hurst parameter $H$'s. The Equation \eqref{eq:5} is motivated by the result for Brownian motion in \citet{rabi2006multiple} when $H = \frac{1}{2}$. We observe there is a certain relation between the optimal threshold $\hat{\eta}$'s and the time horizon $T$ for different Hurst parameter $H$'s. We use this pattern to model the optimal thresholds for the multiple samples case.

\subsection{Multiple samples case}
Similar to the deterministic sampling, we generalize the discussion above in the case of multiple samples. Therefore, given the thresholds $\eta_1 \ge 0, \eta_2 \ge 0, \ldots, \eta_N \ge 0$,  we define the sampling times $\tau_{\eta_1} = \inf_{t \ge 0}\{t: |B^H_t|\ge \eta_1\}$ and  $\tau_{\eta_{i}} = \inf_{t \ge \tau_{i-1}}\{t: |B^H_t - B^H_{\tau_{i-1}}|\ge \eta_{i}\}$ for $2 \le i \le N$. The actual sampling time is given by $\tau_i = \tau_{\eta_i} \land T$ for $1 \le i \le N$.  The performance of the estimation is again quantified using the squared error distortion:
\begin{align*}
J(\eta_1, \eta_2, \ldots, \eta_N) &= \mathbb{E}[\sum_{i = 0}^N\int_{\tau_i}^{\tau_{i+1}}(B_t^H - \hat{B}_t^H)^2dt]\\
&= \frac{T^{2H+1}}{2H+1} - 2 \mathbb{E} \sum_{i = 1}^N \int_{\tau_i}^{\tau_{i+1}} B_t^H\hat{B}_t^Hdt + \mathbb{E} \sum_{i = 1}^N   \int_{\tau_i}^{\tau_{i+1}} (\hat{B}_t^H)^2dt
\end{align*}
where the estimation $\hat{B}_t^H$ is the same as in Equation \eqref{eq:2}. Considering that the numerical computation of $\mathbb{E} \int_{\tau_i}^{\tau_{i+1}} B_t^H\hat{B}_t^Hdt$ using simulation requires several sample paths and keeping in mind that this computation has to be performed at every iteration during the optimization algorithm, it is not feasible to directly optimize the distortion in real time. Therefore, we use the estimation $\hat{B}_t^H$ to represent the true $B_t^H$ in $\mathbb{E} \int_{\tau_i}^{\tau_{i+1}} B_t^H\hat{B}_t^Hdt$ and propose a computationally feasible  modified distortion function, according to which
\begin{align}
J^*(\eta_1, \eta_2, \ldots, \eta_N) = \frac{T^{2H+1}}{2H+1}- \mathbb{E}  \sum_{i = 1}^N  \int_{\tau_i}^{\tau_{i+1}} (\hat{B}_t^H)^2dt  = \frac{T^{2H+1}}{2H+1}- \mathbb{E}h^*((\eta_1, \eta_2, \ldots, \eta_N), B_t^H),
\label{eq:modified-distortion}
\end{align}
where $h^*((\eta_1, \eta_2, \ldots, \eta_N), B_t^H)$ is the observation with variation at the value $(\eta_1, \eta_2, \ldots, \eta_N)$ calculated using a fBm sample path $B_t^H$. We can see the modified distortion in Equation \eqref{eq:modified-distortion} is still not an explicit form and it still needs to be further analyzed to obtain the optimal thresholds. Therefore, we still investigate a truncated version of the problem like in the deterministic sampling. Instead of using all previous sample information like in Equation \eqref{eq:2}, we estimate $\hat{B}_t^H$ using only the last observed value. Therefore,  the estimation of $\hat{B}_t$ at time $\tau_i \le t < \tau_{i+1}$, becomes:
\begin{align*}
\hat{B}_t^H &=  \mathbb{E}(B_t^H|B^H_{\tau_i}) = \frac{B_{\tau_i}^H}{2\tau_i^{2H}}(t^{2H}+\tau_i^{2H}-|t-\tau_i|^{2H}).\\
\end{align*}
This is the truncated case for which we obtain the truncated formula for the modified distortion.

\begin{prop}
\label{prop-two-sample-trunc}
When the underlying process is a fractional Brownian motion with Hurst parameter $H\in(0,1)$, the modified aggregate squared error distortion for $N$ samples case with thresholds $\eta_1, \eta_2, \ldots, \eta_N$ with truncated estimation is obtained by
\begin{align}
J^*(\eta_1, \eta_2, \ldots, \eta_N) &=   \frac{T^{2H+1}}{2H+1}-\mathbb{E}h^*((\eta_1, \eta_2, \ldots, \eta_N), B_t^H) = \frac{T^{2H+1}}{2H+1}- \mathbb{E} \sum_{i = 1}^N   \int_{\tau_i}^{\tau_{i+1}} (\hat{B}_t^H)^2dt\\
&= \frac{T^{2H+1}}{2H+1} -  \mathbb{E}  \sum_{i = 1}^{N} \int_{\tau_i}^{\tau_{i+1}}(\frac{B_{\tau_i}^H}{2\tau_i^{2H}}(t^{2H}+\tau_i^{2H}-|\tau_i-t|^{2H}))^2dt. \label{two-sample-trunc}
\end{align}
\end{prop}
We still investigate this version of the problem and compare the results with the non-truncated case, since there is a trade-off between computational complexity and accuracy.

From Section \ref{section:relationship} and Equation \eqref{eq:5}, we can see given a Hurst parameter, the optimal thresholds depend on the time horizon $T$. This has also been shown for Brownian motion in \citet{rabi2006multiple} when $H = \frac{1}{2}$.  Then for the multiple samples case, the threshold $\eta_i$, for $1 < i \le N$, depends on the remaining time horizon $T - \tau_{i-1}$. Since $\tau_{i-1}$ is stochastic and depends on the actual fBm trajectory,  it is not feasible to directly optimize the distortion based on $\eta_i$'s. Considering the relation between optimal threshold and the time horizon we observed in Equation \eqref{eq:5}, we propose to find the optimal threshold $\eta_i$'s for $i = 1, 2, \ldots, N$ having the following structure:
\begin{align*}
\eta_1 &= q_1 T^{\frac{1}{4}(2H+1)},\\
\eta_2 &= q_2 (T-\tau_1)^{\frac{1}{4}(2H+1)},\\
\vdots \\
\eta_N &= q_N (T-\tau_{N-1})^{\frac{1}{4}(2H+1)}.
\end{align*}
Therefore, in order to find the optimal $\hat{\eta}_i$'s under our framework, we need to find the optimal $\hat{q}_i$'s, which lead to writing the modified distortion $J^*(\eta_1, \eta_2, \ldots, \eta_N)$ as a function of the $q_i$'s, i.e. $J^*(q_1, q_2, \ldots, q_N)$. We also write $h^*((q_1, q_2, \ldots, q_N), B_t^H)$ as the equivalent $h^*((\eta_1, \eta_2, \ldots, \eta_N), B_t^H)$. In order to optimize $J^*(q_1, q_2, \ldots, q_N)$, we need to use a multidimensional version of the  stochastic approximation method. 

\subsubsection{Multidimensional stochastic optimization}
A multidimensional version of the KW-algorithm was introduced by \citet{blum1954}. The multidimensional KW-algorithm uses a one-sided finite-difference approximation of the gradient in every direction and can be described by the following way:
\begin{align}
q^{(n+1)}_{k} = q^{(n)}_k + a^{(n)}_k Y^{(n)}_k,
\end{align}
where $k$ represents the $k$th direction in $(q_1, q_2, \ldots, q_N)$, $n$ represents the iteration number and $Y^{(n)}_{k} = [h^*(q^{(n)} + c^{(n)}_ke_k, \delta_{k}^{(n)+})-h^*(q^{(n)}, \delta_{k}^{(n)-})]/c^{(n)}_k$, where $e_k$ is the standard basis in $N$ dimensional real number space for the $k$th dimension. The $\delta_{k}^{(n)+}$'s and $\delta_{k}^{(n)-}$'s are all independent fBm sample paths. Since using only two independent fBm sample paths for each direction $k$ at each iteration is too noisy and has limited finite time convergence results, we use multiple independent fBm sample paths and take an average of them at each iteration. This can give us good finite convergence results. The tuning sequences $a^{(n)}_k$ and $c^{(n)}_k$ are all deterministic. \citet{broadie2014} proposed an adaptive version of the multidimensional KW algorithm that can improve the finite-time performance.  They proposed the optimal choice of the tuning sequences as $a^{(n)}_k = \alpha_k/(n + \beta_k)$'s and $c^{(n)}_k = \gamma_k/n^{1/4}$'s for some $\alpha_k, \gamma_k \in \mathbb{R}_+ \text{ and } \beta_k \in \mathbb{Z}_+$. Their adaptive version of the method dynamically scales and shifts the tuning sequences $a_k^{(n)}$'s and $c_k^{(n)}$'s by adaptively adjusting $\alpha_k$'s, $\beta_k$'s and $\gamma_k$'s based on the characteristics and noise level of the function that needs to be optimized. Their modification can achieve better finite-time performance. For more details, the reader can refer to \citet{broadie2014}. We use this adaptive version of the multidimensional KW algorithm to find the optimal $\hat{q}_1, \hat{q}_2, \ldots, \hat{q}_N$ and obtain the corresponding optimal $\hat{\eta}_1, \hat{\eta}_2, \ldots, \hat{\eta}_N$. In the following sections, we  use $N = 2$ and $N=3$ as examples to show how to find the optimal sets of thresholds.

\subsection{Computation of the optimal thresholds}

In this section, we use $N = 2$ and $N = 3$ as an example to illustrate how to find the optimal sets of thresholds $\hat{\eta}_1, \hat{\eta}_2, \ldots, \hat{\eta}_N$ in both truncated sampling case and non-truncated sampling case.

\subsubsection{Two samples case}

For $\tau_2 \le t < T$, given $\tau_{\eta_1} = \inf_{t \ge 0}\{t: |B^H_t|\ge \eta_1\}$,  $\tau_{\eta_{2}} = \inf_{t \ge \tau_{1}}\{t: |B^H_t - B^H_{\tau_{1}}|\ge \eta_{2}\}$ and $\tau_i = \tau_{\eta_i} \land T$ for $i = 1, 2$,  based on Equation \eqref{eq:2}, the estimation $\hat{B}_t^H$ is
\begin{align*}
\hat{B}_t^H &= C(\tau_1, \tau_2, H)(A_1(\tau_1, \tau_2, t, H)B_{\tau_1}^H - A_2(\tau_1, \tau_2, t, H)B_{\tau_2}^H),
\end{align*}
where $C(\tau_1, \tau_2, H)$, $A_1(\tau_1, \tau_2, t, H)$ and $A_2(\tau_1, \tau_2, t, H)$ are given in Section \ref{appendix:corollary1}. 

\begin{cor}
\label{prop-level2}
When the underlying process is a fractional Brownian motion with Hurst parameter $H\in(0,1)$, the modified aggregate squared error distortion for level-triggered sampling two samples case with thresholds $\eta_1, \eta_2$ is:
\begin{align*}
J^*(\eta_1, \eta_2) &= \frac{T^{2H+1}}{2H+1} \\
&- \mathbb{E}\bigg[\frac{(B_{\tau_1}^H)^2}{4}\big[(\tau_2-\tau_{1}) + \frac{\tau_2^{4H+1}/\tau_1^{4H}-\tau_1}{4H+1} + \frac{2(\tau_2^{2H+1}-\tau_1^{2H+1})/\tau_1^{2H}}{2H+1} \\
& +\frac{(\tau_2 - \tau_1)^{4H+1}/\tau_1^{4H}}{4H+1} - 2\frac{(\tau_2 - \tau_1)^{2H+1}/\tau_1^{2H}}{2H+1} - 2\tau_1\int_1^{\frac{\tau_2}{\tau_1}}\mu^{2H}(\mu-1)^{2H}d\mu
\big] \\
&+ C(\tau_1, \tau_2, H)^2\big[(B_{\tau_1}^H)^2 \int_{\tau_2}^T A_1^2(\tau_1, \tau_2, t, H)dt + (B_{\tau_2}^H)^2 \int_{\tau_2}^T A_2^2(\tau_1, \tau_2, t, H)dt\\
& - 2B^H_{\tau_1}B^H_{\tau_2}\int_{\tau_2}^T A_1(\tau_1, \tau_2, t, H)A_2(\tau_1, \tau_2, t, H)dt\big]\bigg]
\end{align*}
where $\tau_{\eta_1} = \inf_{t \ge 0}\{t: |B^H_t|\ge \eta_1\}$,  $\tau_{\eta_{2}} = \inf_{t \ge \tau_{1}}\{t: |B^H_t - B^H_{\tau_{1}}|\ge \eta_{2}\}$ and $\tau_i = \tau_{\eta_i} \land T$ for $i = 1, 2$.
\end{cor}
The proof of Corollary \ref{prop-level2} is shown in Section \ref{appendix:prop6}. Using Broadie method \citep{broadie2014}, we obtain the optimal thresholds $\hat{\eta}_1$ and $\hat{\eta}_2$ that minimize $J^*(\eta_1, \eta_2)$ in Corollary \ref{prop-level2}. For $H =0.4, 0.6 \text{ and } 0.7$, the optimal thresholds $(\hat{\eta}_1, \hat{\eta}_2)$ and corresponding minimized distortion is given in Table \ref{table-level-two-samples-non}. Using the same method, we can also obtain the optimal thresholds $\hat{\eta}_1$ and $\hat{\eta}_2$ for the truncated level-triggered sampling by minimizing the distortion in Proposition \ref{prop-two-sample-trunc} with $N = 2$. The corresponding optimal thresholds $(\hat{\eta}_1, \hat{\eta}_2)$ and the minimized distortion are given in Table \ref{table-level-two-samples-trunc}. We also show the results of the minimized distortion given by the truncated and non-truncated two samples case in Figure \ref{level_two_loss_compa}. For $H = 0.1, 0.2 \text{ and } 0.3$, we find including the second allowed sample can not help to improve the estimation accuracy and decrease the distortion. As a result, we do not include the results here.

\begin{table}[h]
\begin{center}
\begin{tabular}{| c | c | c |}
\hline
Hurst & $(\hat{\eta}_1, \hat{\eta}_2)$ & Distortion \\
\hline
0.4 & $(0.693T^{\frac{1}{4}(2H+1)}, 0.991(T-\tau_1)^{\frac{1}{4}(2H+1)})$ & 35.093   \\
\hline
0.6 & $(0.658T^{\frac{1}{4}(2H+1)}, 0.930(T-\tau_1)^{\frac{1}{4}(2H+1)})$ & 58.133   \\
\hline
0.7 & $(0.627T^{\frac{1}{4}(2H+1)}, 1.023(T-\tau_1)^{\frac{1}{4}(2H+1)})$ & 69.997 \\
\hline
\end{tabular}
\caption{\label{table-level-two-samples-non}Optimal thresholds $(\hat{\eta}_1, \hat{\eta}_2)$ and corresponding minimized distortions in the two samples non-truncated case for Hurst parameter $H = 0.4, 0.6,  0.7$ when $T = 20$.}
\end{center}
\end{table}


\begin{table}[h]
\begin{center}
\begin{tabular}{| c | c | c |}
\hline
Hurst & $(\hat{\eta}_1, \hat{\eta}_2)$ & Distortion \\
\hline
0.4 & $(0.666T^{\frac{1}{4}(2H+1)}, 1.013(T-\tau_1)^{\frac{1}{4}(2H+1)})$ & 35.820   \\
\hline
0.6 & $(0.668T^{\frac{1}{4}(2H+1)}, 0.941(T-\tau_1)^{\frac{1}{4}(2H+1)})$ & 59.582   \\
\hline
0.7 & $(0.649T^{\frac{1}{4}(2H+1)}, 1.039(T-\tau_1)^{\frac{1}{4}(2H+1)})$ & 77.191  \\
\hline
\end{tabular}
\caption{\label{table-level-two-samples-trunc}Optimal thresholds $(\hat{\eta}_1, \hat{\eta}_2)$ and corresponding minimized distortions in the two truncated samples case for Hurst parameter $H = 0.4, 0.6,  0.7$ when $T = 20$.}
\end{center}
\end{table}



\begin{figure}[h]
\centering
  \includegraphics[width=0.5\linewidth]{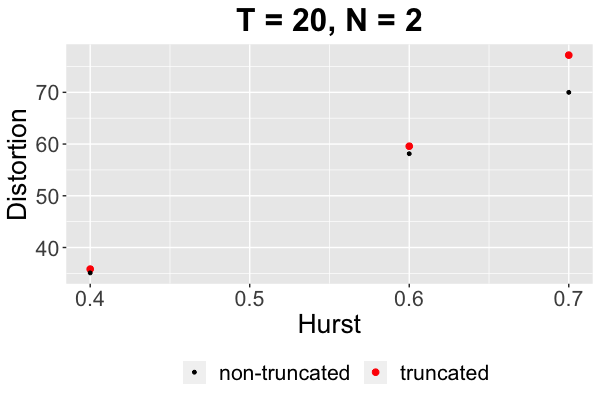}
  \caption{Optimal distortions given the optimal thresholds  for $H = 0.4, 0.6, 0.7$ when $T = 20$ for the two samples case. The black point represents the optimal distortion for the non-truncated two samples case for the corresponding Hurst parameter. The red point represents the optimal distortion for the truncated two samples case for the corresponding Hurst parameter.}
  \label{level_two_loss_compa}
\end{figure}

Based on the results in Table \ref{table-level-two-samples-non} and Table \ref{table-level-two-samples-trunc}, we can see the estimated optimal $\hat{q}_1$ and $\hat{q}_2$ have similar values for different $H$'s. Considering $T$ and $(T - \tau_{1})$ are usually some values that are bigger than 1, this shows, as the Hurst parameter $H$ increases, the optimal thresholds become larger. This matches the same trend as that shown in the one sample case. Also, non-truncated and truncated two samples case tend to have similar optimal thresholds. The non-truncated two samples case tends to have slightly smaller optimal distortion.

\subsubsection{Three samples case}

For $\tau_3 \le t < T$, given $\tau_{\eta_1} = \inf_{t \ge 0}\{t: |B^H_t|\ge \eta_1\}$,  $\tau_{\eta_{i}} = \inf_{t \ge \tau_{i-1}}\{t: |B^H_t - B^H_{\tau_{i-1}}|\ge \eta_{i}\}$ for $i=2, 3$ and $\tau_i = \tau_{\eta_i} \land T$ for $i = 1, 2, 3$, the estimation is
\begin{align*}
\hat{B}_t^H = C(\tau_1, \tau_2, \tau_3, H)(A_1(\tau_1, \tau_2, \tau_3, t, H)B_{\tau_1}^H + A_2(\tau_1, \tau_2, \tau_3, t, H)B_{\tau_2}^H  + A_3(\tau_1, \tau_2, \tau_3, t, H)B_{\tau_3}^H),
\end{align*}
where $C(\tau_1, \tau_2, \tau_3, H)$, $A_1(\tau_1, \tau_2, \tau_3, t, H)$, $A_2(\tau_1, \tau_2, \tau_3, t, H)$ and $A_3(\tau_1, \tau_2, \tau_3, t, H)$ are given in Section \ref{appendix:corollary2}.

\begin{cor}
\label{prop-level3}
When the underlying process is a fractional Brownian motion with Hurst parameter $H\in(0,1)$, the modified aggregate squared error distortion for level-triggered sampling three samples case with thresholds $\eta_1, \eta_2, \eta_3$ is:
\label{equ:three_non_dis}
\begin{align*}
&J^*(\eta_1, \eta_2, \eta_3) = \frac{T^{2H+1}}{2H+1} \\
&- \mathbb{E}\bigg[\frac{(B_{\tau_1}^H)^2}{4} \big((\tau_2-\tau_{1}) + \frac{\tau_2^{4H+1}/\tau_1^{4H}-\tau_1}{4H+1} + \frac{2(\tau_2^{2H+1}-\tau_1^{2H+1})/\tau_1^{2H}}{2H+1} \\
& +\frac{(\tau_2 - \tau_1)^{4H+1}/\tau_1^{4H}}{4H+1} - 2\frac{(\tau_2 - \tau_1)^{2H+1}/\tau_1^{2H}}{2H+1} - 2\tau_1\int_1^{\frac{\tau_2}{\tau_1}}\mu^{2H}(\mu-1)^{2H}d\mu \big)
 \\
&+ C(\tau_1, \tau_2, H)^2\big[(B_{\tau_1}^H)^2 \int_{\tau_2}^{\tau_3} A_1^2(\tau_1, \tau_2, t, H)dt + (B_{\tau_2}^H)^2 \int_{\tau_2}^{\tau_3} A_2^2(\tau_1, \tau_2, t, H)dt\\
& - 2B^H_{\tau_1}B^H_{\tau_2}\int_{\tau_2}^{\tau_3} A_1(\tau_1, \tau_2, t, H)A_2(\tau_1, \tau_2, t, H)dt \big]  \\
& + C(\tau_1, \tau_2, \tau_3, H)^2 \big[ (B_{\tau_1}^H)^2 \int_{\tau_3}^{T} A_1(\tau_1, \tau_2, \tau_3, t, H)^2 dt + (B_{\tau_2}^H)^2 \int_{\tau_3}^{T} A_2(\tau_1, \tau_2, \tau_3, t, H)^2 dt \\
&+ (B_{\tau_3}^H)^2 \int_{\tau_3}^{T} A_3(\tau_1, \tau_2, \tau_3, t, H)^2 dt + 2B_{\tau_1}^HB_{\tau_2}^H  \int_{\tau_3}^{T} A_1(\tau_1, \tau_2, \tau_3, t, H)A_2(\tau_1, \tau_2, \tau_3, t, H) dt \\
&+ 2B_{\tau_1}^HB_{\tau_3}^H \int_{\tau_3}^{T} A_1(\tau_1, \tau_2, \tau_3, t, H)A_3(\tau_1, \tau_2, \tau_3, t, H) dt  \\
&+ 2B_{\tau_2}^HB_{\tau_3}^H \int_{\tau_3}^{T} A_2(\tau_1, \tau_2, \tau_3, t, H)A_3(\tau_1, \tau_2, \tau_3, t, H) dt \big]\bigg],
\end{align*}
where $\tau_{\eta_1} = \inf_{t \ge 0}\{t: |B^H_t|\ge \eta_1\}$,  $\tau_{\eta_{i}} = \inf_{t \ge \tau_{i-1}}\{t: |B^H_t - B^H_{\tau_{i-1}}|\ge \eta_{i}\}$ for $i=2, 3$ and $\tau_i = \tau_{\eta_i} \land T$ for $i = 1, 2, 3$.
\end{cor}
The proof of Corollary \ref{prop-level3} is shown in Section \ref{appendix:prop7}. Using Broadie's method \citep{broadie2014}, we find the optimal thresholds that minimize $J^*(\eta_1, \eta_2, \eta_3)$ in Corollary \ref{equ:three_non_dis}. For $H =0.1, 0.2, \ldots, 0.7$, the optimal thresholds $(\hat{\eta}_1, \hat{\eta}_2, \hat{\eta}_3)$ and the corresponding minimized distortion are given in Table \ref{table-level-three-samples-non}. Using the same method, we also obtain the optimal thresholds for the truncated level-triggered sampling by minimizing the distortion in Proposition \ref{prop-two-sample-trunc} with $N = 3$. The corresponding optimal thresholds $(\hat{\eta}_1, \hat{\eta}_2, \hat{\eta}_3)$ and minimized distortion are given in Table \ref{table-level-three-samples-trunc}. We show the results of the minimized distortion given by the truncated and non-truncated three samples case in Figure \ref{level_three_loss_comp}.

\begin{table}[h]
\begin{center}
\begin{tabular}{| c | c | c |}
\hline
Hurst & $(\hat{\eta}_1, \hat{\eta}_2, \hat{\eta}_3)$ & Distortion \\
\hline
0.1 & $(0.775T^{\frac{1}{4}(2H+1)}, 0.617(T-\tau_1)^{\frac{1}{4}(2H+1)}, 1.051(T-\tau_2)^{\frac{1}{4}(2H+1)})$ & 2.259   \\
\hline
0.2 & $(0.737T^{\frac{1}{4}(2H+1)}, 0.534(T-\tau_1)^{\frac{1}{4}(2H+1)}, 0.982(T-\tau_2)^{\frac{1}{4}(2H+1)})$ & 8.049   \\
\hline
0.3 & $(0.719T^{\frac{1}{4}(2H+1)}, 0.492(T-\tau_1)^{\frac{1}{4}(2H+1)}, 0.936(T-\tau_2)^{\frac{1}{4}(2H+1)})$ & 18.388   \\
\hline
0.4 & $(0.646T^{\frac{1}{4}(2H+1)}, 0.655(T-\tau_1)^{\frac{1}{4}(2H+1)}, 0.939(T-\tau_2)^{\frac{1}{4}(2H+1)})$ & 28.713   \\
\hline
0.6 & $(0.542T^{\frac{1}{4}(2H+1)}, 0.733(T-\tau_1)^{\frac{1}{4}(2H+1)}, 0.935(T-\tau_2)^{\frac{1}{4}(2H+1)})$ & 22.766   \\
\hline
0.7 & $(0.485T^{\frac{1}{4}(2H+1)}, 0.760(T-\tau_1)^{\frac{1}{4}(2H+1)}, 1.012(T-\tau_2)^{\frac{1}{4}(2H+1)})$ & 1.661   \\
\hline
\end{tabular}
\caption{\label{table-level-three-samples-non}Optimal thresholds $(\hat{\eta}_1, \hat{\eta}_2, \hat{\eta}_3)$ and corresponding distortions in the three samples non-truncated case for  Hurst parameter $H = 0.1, 0.2, \ldots, 0.7$ when $T = 20$.}
\end{center}
\end{table}

\begin{table}[h]
\begin{center}
\begin{tabular}{| c | c | c |}
\hline
Hurst & $(\hat{\eta}_1, \hat{\eta}_2, \hat{\eta}_3)$ & Distortion \\
\hline
0.1 & $(0.682T^{\frac{1}{4}(2H+1)}, 0.545(T-\tau_1)^{\frac{1}{4}(2H+1)},  0.983(T-\tau_2)^{\frac{1}{4}(2H+1)})$ & 0.877   \\
\hline
0.2 & $(0.629T^{\frac{1}{4}(2H+1)}, 0.489(T-\tau_1)^{\frac{1}{4}(2H+1)}, 0.927(T-\tau_2)^{\frac{1}{4}(2H+1)})$ & 8.437   \\
\hline
0.3 & $(0.616T^{\frac{1}{4}(2H+1)}, 0.513(T-\tau_1)^{\frac{1}{4}(2H+1)}, 0.918(T-\tau_2)^{\frac{1}{4}(2H+1)})$ & 20.916   \\
\hline
0.4 & $(0.579T^{\frac{1}{4}(2H+1)}, 0.647(T-\tau_1)^{\frac{1}{4}(2H+1)}, 0.957(T-\tau_2)^{\frac{1}{4}(2H+1)})$ & 29.020   \\
\hline
0.6 & $(0.529T^{\frac{1}{4}(2H+1)}, 0.705(T-\tau_1)^{\frac{1}{4}(2H+1)}, 0.971(T-\tau_2)^{\frac{1}{4}(2H+1)})$ & 25.803   \\
\hline
0.7 & $(0.479T^{\frac{1}{4}(2H+1)}, 0.756(T-\tau_1)^{\frac{1}{4}(2H+1)}, 1.069(T-\tau_2)^{\frac{1}{4}(2H+1)})$ & 6.553   \\
\hline
\end{tabular}
\caption{\label{table-level-three-samples-trunc}Optimal thresholds $(\hat{\eta}_1, \hat{\eta}_2, \hat{\eta}_3)$ and corresponding distortions in the three samples truncated case for  Hurst parameter $H = 0.1, 0.2, \ldots, 0.7$ when $T = 20$.}
\end{center}
\end{table}

\begin{figure}[h]
\centering
  \includegraphics[width=0.5\linewidth]{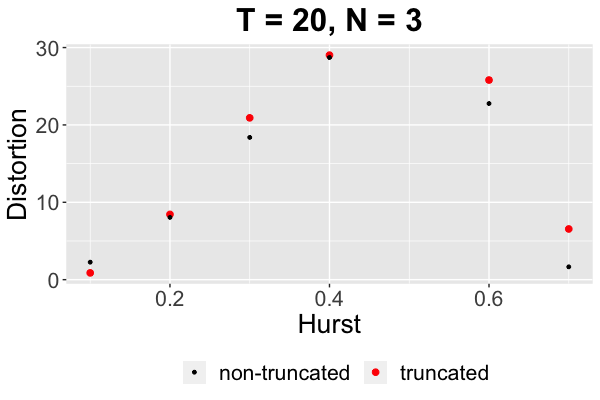}
  \caption{Optimal distortions given the optimal thresholds for $H = 0.1, 0.2, \ldots, 0.7$ when $T = 20$. The black point represents the optimal distortion for the non-truncated three samples case with the corresponding Hurst parameter. The red point represents the optimal distortion for the truncated three samples case with the corresponding Hurst parameter.}
  \label{level_three_loss_comp}
\end{figure}

From Tables \ref{table-level-three-samples-non} and \ref{table-level-three-samples-trunc}, we can see, for different $H$, the estimated optimal $\hat{q}_i$ does vary. For $\hat{\eta}_2$ and $\hat{\eta}_3$, $T - \tau_{2}$ and $T - \tau_{3}$ are values that are bigger than 1, and $\hat{q}_2$'s and $\hat{q}_3$'s are similar for different $H$'s as $H$ increases. These show that, as the Hurst parameter $H$ increases, the optimal thresholds $\hat{\eta}_2$ and $\hat{\eta}_3$ become larger. For $\hat{\eta}_1$, even though $\hat{q}_1$ has a decreasing trend as $H$ increases, but the the dominant part in $\hat{\eta}_1$ is the $T^{\frac{1}{4}(2H+1)}$ which includes a whole time horizon. This shows that, as the Hurst parameter $H$ increases, the optimal threshold $\hat{\eta}_1$ increases. This matches the same trend as those in the one sample and two samples case. Also, for non-truncated and truncated case, they tend to have similar optimal thresholds.  For the minimized distortion, the non-truncated case has a smaller distortion in most cases. 

\section{Discussion}
\label{section:conclusion}

In this paper, we propose two different classes of sampling schemes for the estimation of the fractional Brownian motion. By minimizing the aggregate squared error distortion, we find the optimal sampling strategy for each sampling scheme when different numbers of samples are allowed. We observe a certain relation between the optimal sampling strategy and the Hurst parameter given a certain number of allowed samples.

For the deterministic sampling, the sampling time sequence is deterministic and independent of the actual trajectory of the fBm process. We can directly optimize the distortion to find the optimal sampling times. Based on our results, when the Hurst parameter keeps further away from $H = \frac{1}{2}$, the optimal sampling times tend to be smaller. We also investigate a truncated version of this problem by doing the estimation using only the last observed value. This truncated version gives similar optimal sampling times and minimized distortions, indicating that it has good accuracy with less computational complexity.

For the level-triggered sampling, the sampling time sequence is stochastic and depends on the actual fBm trajectory. We use the stochastic approximation method to find the optimal thresholds that minimize the distortion. For the multiple samples case, it is not feasible to directly find the optimal thresholds. We estimate the optimal thresholds with a proposed structure. Based on our results, when the Hurst parameter keeps increasing, the optimal thresholds have an increasing trend. We still investigate a truncated version of this problem by doing the estimation using only the last observed value. We find this truncated version still gives similar optimal thresholds and distortions compared with the non-truncated version. This shows that using this truncated version estimation for level-triggered sampling can still have a good performance with less computational complexity.

In this paper, we focus on the signal to be the fractional Brownian motion. A future direction is to consider the multiple sampling problems for the fractional Ornstein-Uhlenbeck process. The fractional Ornstein-Uhlenbeck process is the fractional analogue of the Ornstein-Uhlenbeck process. 
The existing multiple sampling framework for the Ornstein-Uhlenbeck process \citep{rabi2006multiple} can not be extended to the fractional Ornstein-Uhlenbeck process directly. It is of interest to extend our work and study the multiple sampling methods for the fractional Ornstein-Uhlenbeck process.

\section{Proofs and Derivations}

\subsection{Proof of Corollary \ref{cor-deter-2}}
\label{appendix:corollary1}
\begin{proof}
For $\tau_1 \le t < \tau_2$, the estimation $\hat{B}_t^H$ is
\begin{align*}
\hat{B}_t^H = \mathbb{E}(B_t^H|\mathcal{F}_{\tau_1}) = \mathbb{E}(B_t^H|B^H_{\tau_1}) = \frac{B_{\tau_1}^H}{2\tau_1^{2H}}(t^{2H}+\tau_1^{2H}-|t-\tau_1|^{2H}).
\end{align*}
For $\tau_2 \le t < T$, given
\begin{align*}
\Sigma & = \begin{pmatrix}
t^{2H} &  \frac{1}{2}(t^{2H}+\tau_1^{2H}- (t-\tau_1)^{2H}) &  \frac{1}{2}(t^{2H}+\tau_2^{2H}- (t-\tau_2)^{2H})\\
\frac{1}{2}(t^{2H}+\tau_1^{2H}- (t-\tau_1)^{2H}) & \tau_1^{2H} & \frac{1}{2}(\tau_1^{2H}+\tau_2^{2H}- (\tau_2-\tau_1)^{2H})\\
 \frac{1}{2}(t^{2H}+\tau_2^{2H}- (t-\tau_2)^{2H}) & \frac{1}{2}(\tau_1^{2H}+\tau_2^{2H}- (\tau_2-\tau_1)^{2H}) & \tau_2^{2H} 
\end{pmatrix}\\
& = \begin{pmatrix}
\Sigma_{11} & \Sigma_{12}\\
\Sigma_{12} & \Sigma_{22}
\end{pmatrix},
\end{align*}
and Equation \eqref{eq:2}, the estimation $\hat{B}_t^H$ is:
\begin{align*}
\hat{B}_t^H &= \Sigma_{12}\Sigma_{22}^{-1}
\begin{pmatrix}
B_{\tau_1}^H\\
B_{\tau_2}^H
\end{pmatrix} \\
=& 
\begin{pmatrix}
 \frac{1}{2}(t^{2H}+\tau_1^{2H}- (t-\tau_1)^{2H}) &  \frac{1}{2}(t^{2H}+\tau_2^{2H}- (t-\tau_2)^{2H})
\end{pmatrix} \cdot \\
&\begin{pmatrix}
 \tau_1^{2H} & \frac{1}{2}(\tau_1^{2H}+\tau_2^{2H}- (\tau_2-\tau_1)^{2H})\\
 \frac{1}{2}(\tau_1^{2H}+\tau_2^{2H}- (\tau_2-\tau_1)^{2H}) & \tau_2^{2H} 
\end{pmatrix}^{-1} 
\begin{pmatrix}
B_{\tau_1}^H\\
B_{\tau_2}^H
\end{pmatrix}\\
&= C(\tau_1, \tau_2, H)(A_1(\tau_1, \tau_2, t, H)B_{\tau_1}^H - A_2(\tau_1, \tau_2, t, H)B_{\tau_2}^H),
\end{align*}
where
\begin{align*}
C(\tau_1, \tau_2, H) &= \frac{1}{\tau_1^{4H}+(\tau_2^{2H} - (-\tau_1 + \tau_2)^{2H})^2 - 2\tau_1^{2H}(\tau_2^{2H} +(-\tau_1+\tau_2)^{2H})},\\
A_1(\tau_1, \tau_2, t, H) & = 2(t-\tau_1)^{2H}\tau_2^{2H} - (t-\tau_2)^{2H}\tau_2^{2H} + \tau_2^{4H}+(t-\tau_2)^{2H}(-\tau_1+\tau_2)^{2H}-\tau_2^{2H}(-\tau_1+\tau_2)^{2H}\\
&-\tau_1^{2H}((t-\tau_2)^{2H}+\tau_2^{2H})+t^{2H}(\tau_1^{2H}-\tau_2^{2H}-(-\tau_1+\tau_2)^{2H}), \\
\text{and} \\
A_2(\tau_1, \tau_2, t, H) &= -\tau_1^{4H} + (t-\tau_1)^{2H}(\tau_2^{2H}-(-\tau_1+\tau_2)^{2H})+t^{2H}(\tau_1^{2H} - \tau_2^{2H}+(-\tau_1+\tau_2)^{2H})\\
&+\tau_1^{2H}((t-\tau_1)^{2H}-2(t-\tau_2)^{2H}+\tau_2^{2H}+(-\tau_1+\tau_2)^{2H}).
\end{align*}
Based on Equation \eqref{eq:dis-non}, the aggregate squared error distortion is
\begin{align*}
J(\tau_1, \tau_2) &= \int_{0}^T\mathbb{E}(B_t^H)^2dt - \sum_{i = 1}^2\int_{\tau_i}^{\tau_{i+1}} \mathbb{E}(\hat{B}_t^H)^2dt\\
&= \frac{T^{2H+1}}{2H+1} - \int_{\tau_1}^{\tau_{2}}\mathbb{E}\bigg(\frac{B_{\tau_1}^H}{2\tau_1^{2H}}(t^{2H}+\tau_1^{2H}-|\tau_1-t|^{2H})\bigg)^2dt\\
& -\int_{\tau_2}^{T}\mathbb{E}\big[C(\tau_1, \tau_2, H)^2[A_1(\tau_1, \tau_2, t, H)B_{\tau_1}^H - A_2(\tau_1, \tau_2, t, H)B_{\tau_2}^H]^2\big]dt,
\end{align*}
where 
\begin{align*}
\int_{\tau_1}^{\tau_{2}}\mathbb{E}\bigg(\frac{B_{\tau_1}^H}{2\tau_1^{2H}}(t^{2H}+\tau_1^{2H}-|\tau_1-t|^{2H})\bigg)^2dt = \frac{1}{4\tau_1^{2H}}\int_{\tau_1}^{\tau_2} (s^{2H}+\tau_1^{2H}-|s-\tau_1|^{2H})^2ds
\end{align*}
and
\begin{align*}
&\int_{\tau_2}^{T}\mathbb{E}\big[C(\tau_1, \tau_2, H)^2[A_1(\tau_1, \tau_2, t, H)B_{\tau_1}^H - A_2(\tau_1, \tau_2, t, H)B_{\tau_2}^H]^2\big]dt\\
& = C(\tau_1, \tau_2, H)^2\int_{\tau_2}^{T}\mathbb{E}[A_1(\tau_1, \tau_2, t, H)^2(B_{\tau_1}^H)^2 + A_2(\tau_1, \tau_2, t, H)^2(B_{\tau_2}^H)^2\\
& - 2A_1(\tau_1, \tau_2, t, H)A_2(\tau_1, \tau_2, t, H)B_{\tau_1}^HB_{\tau_2}^H]dt\\
& = C(\tau_1, \tau_2, H)^2\int_{\tau_2}^{T}[A_1(\tau_1, \tau_2, t, H)^2\tau_1^{2H} + A_2(\tau_1, \tau_2, t, H)^2\tau_2^{2H}\\
& - A_1(\tau_1, \tau_2, t, H)A_2(\tau_1, \tau_2, t, H)(\tau_2^{2H}+\tau_1^{2H}- |\tau_2-\tau_1|^{2H})]dt.
\end{align*}
\end{proof}

\subsection{Proof of Corollary \ref{cor-deter-3}}
\label{appendix:corollary2}
\begin{proof}
For $0 \le t < \tau_1$ and $\tau_1 \le t < \tau_2$, the estimation $\hat{B}_t^H$ is given in Section \ref{appendix:corollary1}. For $\tau_3 \le t < T$, the estimation $\hat{B}_t^H$ is
\begin{align}
\hat{B}_t^H  = C(\tau_1, \tau_2, \tau_3, H)(A_1(\tau_1, \tau_2, \tau_3, t, H)B_{\tau_1}^H + A_2(\tau_1, \tau_2, \tau_3, t, H)B_{\tau_2}^H  + A_3(\tau_1, \tau_2, \tau_3, t, H)B_{\tau_3}^H) \label{three-samples-level-sample}
\end{align}
where
\begin{align*}
&C(\tau_1, \tau_2, \tau_3, H)=\\
& 1/(2(\tau_2^{4H}(-\tau_1 + \tau_3)^{2H} 
+(-\tau_1 + \tau_2)^{2H}(\tau_3^{4H}+ (-\tau_1 + \tau_3)^{2H}(-\tau_2 + \tau_3)^{2H} \\
&+\tau_3^{2H}((-\tau_1 +\tau_2)^{2H} - (-\tau_1 + \tau_3)^{2H} - (-\tau_2 + \tau_3)^{2H})) + \tau_1^{4H}(-\tau_2 + \tau_3)^{2H}\\
&+\tau_1^{2H}(\tau_2^{2H}((-\tau_1 +\tau_2)^{2H} - (-\tau_1 + \tau_3)^{2H} -(-\tau_2 + \tau_3)^{2H}) \\
&- (-\tau_2 + \tau_3)^{2H}((-\tau_1 + \tau_2)^{2H} + (-\tau_1 + \tau_3)^{2H}- (-\tau_2 + \tau_3)^{2H})\\
& - \tau_3^{2H}((-\tau_1 + \tau_2)^{2H}-(-\tau_1 + \tau_3)^{2H} + (-\tau_2 + \tau_3)^{2H})) \\
&- \tau_2^{2H}(\tau_3^{2H}((-\tau_1 + \tau_2)^{2H}
   + (-\tau_1 + \tau_3)^{2H} - (-\tau_2 + \tau_3)^{2H}) \\
&   + (-\tau_1 + \tau_3)^{2H}((-\tau_1 + \tau_2)^{2H} 
- (-\tau_1 + \tau_3)^{2H} + (-\tau_2 + \tau_3)^{2H})))),
\end{align*}
\begin{align*}
& A_1(\tau_1, \tau_2, \tau_3, t, H) =\\
& -t^{2H}(-\tau_1 + \tau_2)^{2H}\tau_3^{2H}+2(t - \tau_2)^{2H}(-\tau_1 + \tau_2)^{2H}\tau_3^{2H} \\
&- (-\tau_1 + \tau_2)^{2H}(t - \tau_3)^{2H}\tau_3^{2H}-(t - \tau_1)^{2H}\tau_3^{4H} + (t -\tau_2)^{2H}\tau_3^{4H} \\
&+ (-\tau_1 +\tau_2)^{2H}\tau_3^{4H}+t^{2H}\tau_3^{2H}(-\tau_1 + \tau_3)^{2H} - (t - \tau_2)^{2H}\tau_3^{2H}(-\tau_1 + \tau_3)^{2H}\\
 & -t^{2H}(-\tau_1 + \tau_2)^{2H}(-\tau_2 + \tau_3)^{2H} + (-\tau_1 + \tau_2)^{2H}(t - \tau_3)^{2H}(-\tau_2 + \tau_3)^{2H}\\
 &-t^{2H}\tau_3^{2H}(-\tau_2 + \tau_3)^{2H} + 2(t - \tau_1)^{2H}\tau_3^{2H}(-\tau_2 + \tau_3)^{2H}-(t - \tau_2)^{2H}\tau_3^{2H}(-\tau_2 + \tau_3)^{2H} \\
 &- (-\tau_1 + \tau_2)^{2H}\tau_3^{2H}(-\tau_2 + \tau_3)^{2H} -t^{2H}(-\tau_1 + \tau_3)^{2H}(-\tau_2 + \tau_3)^{2H} \\
 &+ (t - \tau_2)^{2H}(-\tau_1 + \tau_3)^{2H}(-\tau_2 + \tau_3)^{2H} +t^{2H}(-\tau_2 + \tau_3)^{4H} \\
 &- (t - \tau_1)^{2H}(-\tau_2 + \tau_3)^{4H}+\tau_2^{4H}(-(t - \tau_1)^{2H} + (t - \tau_3)^{2H}+(-\tau_1 + \tau_3)^{2H}) \\
 &+ \tau_1^{2H}(\tau_2^{2H}((t - \tau_2)^{2H} - (t - \tau_3)^{2H} - (-\tau_2 + \tau_3)^{2H}) + (-\tau_2 + \tau_3)^{2H}(2t^{2H}\\
 &-(t - \tau_2)^{2H} - (t -\tau_3)^{2H} + (-\tau_2 + \tau_3)^{2H}) - \tau_3^{2H}((t - \tau_2)^{2H} - (t - \tau_3)^{2H}\\
 &+(-\tau_2 + \tau_3)^{2H})) - \tau_2^{2H}((-\tau_1 + \tau_2)^{2H}(t - \tau_3)^{2H} + (t - \tau_2)^{2H}(-\tau_1 + \tau_3)^{2H}\\
 &-2(t -\tau_3)^{2H}(-\tau_1 + \tau_3)^{2H} - 2(t - \tau_1)^{2H}(-\tau_2 + \tau_3)^{2H} + (t - \tau_3)^{2H}(-\tau_2 + \tau_3)^{2H} \\
 &+(-\tau_1 + \tau_3)^{2H}(-\tau_2 + \tau_3)^{2H} + \tau_3^{2H}(-2(t - \tau_1)^{2H} + (t - \tau_2)^{2H} + (-\tau_1 + \tau_2)^{2H} \\
 &+(t - \tau_3)^{2H} + (-\tau_1 + \tau_3)^{2H} - 2(-\tau_2 + \tau_3)^{2H}) + t^{2H}(-(-\tau_1 +\tau_2)^{2H} + (-\tau_1 + \tau_3)^{2H} + (-\tau_2 + \tau_3)^{2H})),
\end{align*}
\begin{align*}
& A_2(\tau_1, \tau_2, \tau_3, t, H) =\\
&-(t^{2H}(-\tau_1 + \tau_2)^{2H}\tau_3^{2H} - 2(t - \tau_1)^{2H}(-\tau_1 + \tau_2)^{2H}\tau_3^{2H} +\\
& (-\tau_1 + \tau_2)^{2H}(t - \tau_3)^{2H}\tau_3^{2H} - (t - \tau_1)^{2H}\tau_3^{4H} + (t - \tau_2)^{2H}\tau_3^{4H} - \\
& (-\tau_1 + \tau_2)^{2H}\tau_3^{4H} + t^{2H}(-\tau_1 + \tau_2)^{2H}(-\tau_1 + \tau_3)^{2H} - (-\tau_1 + \tau_2)^{2H}(t - \tau_3)^{2H}(-\tau_1 +\tau_3)^{2H} + \\
&t^{2H}\tau_3^{2H}(-\tau_1 + \tau_3)^{2H} + (t - \tau_1)^{2H}\tau_3^{2H}(-\tau_1 + \tau_3)^{2H} - \\
&2(t - \tau_2)^{2H}\tau_3^{2H}(-\tau_1 + \tau_3)^{2H} + (-\tau_1 + \tau_2)^{2H}\tau_3^{2H}(-\tau_1 + \tau_3)^{2H} -\\
&t^{2H}(-\tau_1 + \tau_3)^{4H} + (t - \tau_2)^{2H}(-\tau_1 + \tau_3)^{4H} -t^{2H}\tau_3^{2H}(-\tau_2 + \tau_3)^{2H} + (t - \tau_1)^{2H}\tau_3^{2H}(-\tau_2 + \tau_3)^{2H} +\\
& t^{2H}(-\tau_1 + \tau_3)^{2H}(-\tau_2 + \tau_3)^{2H}- (t - \tau_1)^{2H}(-\tau_1 + \tau_3)^{2H}(-\tau_2 + \tau_3)^{2H} + \tau_1^{4H}((t - \tau_2)^{2H} -\\ 
& (t - \tau_3)^{2H} - (-\tau_2 + \tau_3)^{2H}) + \tau_2^{2H}((-\tau_1 + \tau_3)^{2H}(-2t^{2H} + (t - \tau_1)^{2H} + (t - \tau_3)^{2H}-\\
& (-\tau_1 + \tau_3)^{2H}) + \tau_3^{2H}((t - \tau_1)^{2H} - (t - \tau_3)^{2H} + (-\tau_1 + \tau_3)^{2H})) + \tau_1^{2H}((-\tau_1 + \tau_2)^{2H}(t - \tau_3)^{2H} +\\
&  (t - \tau_1)^{2H}\tau_3^{2H} - 2(t - \tau_2)^{2H}\tau_3^{2H} + (-\tau_1 + \tau_2)^{2H}\tau_3^{2H} + (t - \tau_3)^{2H}\tau_3^{2H} - \\
& 2(t - \tau_2)^{2H}(-\tau_1 + \tau_3)^{2H} + (t - \tau_3)^{2H}(-\tau_1 + \tau_3)^{2H} - 2\tau_3^{2H}(-\tau_1 + \tau_3)^{2H} + (t - \tau_1)^{2H}(-\tau_2 + \tau_3)^{2H} - \\
&  2(t - \tau_3)^{2H}(-\tau_2 + \tau_3)^{2H} + \tau_3^{2H}(-\tau_2 + \tau_3)^{2H} + (-\tau_1 + \tau_3)^{2H}(-\tau_2 + \tau_3)^{2H} + \tau_2^{2H}(-(t - \tau_1)^{2H} + \\
&(t - \tau_3)^{2H} + (-\tau_1 + \tau_3)^{2H}) + t^{2H}(-(-\tau_1 + \tau_2)^{2H} + (-\tau_1 + \tau_3)^{2H} + (-\tau_2 +\tau_3)^{2H}))),
\end{align*}
\begin{align*}
& A_3(\tau_1, \tau_2, \tau_3, t, H) =\\
&\tau_2^{4H}((t - \tau_1)^{2H} - (t - \tau_3)^{2H} + (-\tau_1 + \tau_3)^{2H}) +\\
&    \tau_1^{4H}((t - \tau_2)^{2H} - (t - \tau_3)^{2H} + (-\tau_2 + \tau_3)^{2H}) + ((-\tau_1 + \tau_2)^{2H})(-(-\tau_1 + \tau_2)^{2H}(t - \tau_3)^{2H}-\\
&    ((t - \tau_1)^{2H} + (t - \tau_2)^{2H} - (-\tau_1 + \tau_2)^{2H})\tau_3^{2H} + (t - \tau_2)^{2H}(-\tau_1 + \tau_3)^{2H} + (t - \tau_1)^{2H}(-\tau_2 + \tau_3)^{2H} + \\
&    t^{2H}((-\tau_1 + \tau_2)^{2H} + 2\tau_3^{2H} - (-\tau_1 + \tau_3)^{2H} - (-\tau_2 + \tau_3)^{2H})) - \tau_2^{2H}((t - \tau_1)^{2H}(-\tau_1 + \tau_2)^{2H} - \\
&    2(-\tau_1 + \tau_2)^{2H}(t - \tau_3)^{2H} + ((t - \tau_1)^{2H} - (t - \tau_2)^{2H} + (-\tau_1 + \tau_2)^{2H})\tau_3^{2H} - \\
&    2(t - \tau_1)^{2H}(-\tau_1 + \tau_3)^{2H} + (t - \tau_2)^{2H}(-\tau_1 + \tau_3)^{2H} + (-\tau_1 + \tau_2)^{2H}(-\tau_1 + \tau_3)^{2H}+\\
&    (t - \tau_1)^{2H}(-\tau_2 + \tau_3)^{2H} + t^{2H}((-\tau_1 + \tau_2)^{2H} + (-\tau_1 + \tau_3)^{2H} - (-\tau_2 + \tau_3)^{2H}))-\\
&    \tau_1^{2H}((t - \tau_2)^{2H}(-\tau_1 + \tau_2)^{2H} -2(-\tau_1 + \tau_2)^{2H}(t - \tau_3)^{2H} - (t - \tau_1)^{2H}\tau_3^{2H} + (t - \tau_2)^{2H}\tau_3^{2H}+\\
&    (-\tau_1 + \tau_2)^{2H}\tau_3^{2H} + (t - \tau_2)^{2H}(-\tau_1 + \tau_3)^{2H} + (t - \tau_1)^{2H}(-\tau_2 + \tau_3)^{2H} - \\
&    2(t - \tau_2)^{2H}(-\tau_2 + \tau_3)^{2H} + (-\tau_1 + \tau_2)^{2H}(-\tau_2 + \tau_3)^{2H} + t^{2H}((-\tau_1 + \tau_2)^{2H} - (-\tau_1 + \tau_3)^{2H}+ \\
&    (-\tau_2 + \tau_3)^{2H}) +\tau_2^{2H}((t - \tau_1)^{2H} + (t - \tau_2)^{2H} - 2(-\tau_1 + \tau_2)^{2H} - 2(t - \tau_3)^{2H} \\
&+ (-\tau_1 + \tau_3)^{2H} +(-\tau_2 + \tau_3)^{2H})).
\end{align*}

Based on Equation \eqref{eq:dis-non}, the aggregate squared error distortion for three samples case is
\begin{align*}
&J(\tau_1, \tau_2, \tau_3) = \int_{0}^T\mathbb{E}(B_t^H)^2dt - \sum_{i = 1}^3\int_{\tau_i}^{\tau_{i+1}} \mathbb{E}(\hat{B}_t^H)^2dt\\
&= \frac{T^{2H+1}}{2H+1} - \int_{\tau_1}^{\tau_{2}}\mathbb{E}\bigg(\frac{B_{\tau_1}^H}{2\tau_1^{2H}}(t^{2H}+\tau_1^{2H}-|\tau_1-t|^{2H})\bigg)^2dt\\
& -\int_{\tau_2}^{\tau_3}\mathbb{E}[C(\tau_1, \tau_2, H)^2[A_1(\tau_1, \tau_2, t, H)B_{\tau_1}^H - A_2(\tau_1, \tau_2, t, H)B_{\tau_2}^H]^2]dt \\
& -\int_{\tau_3}^{T}\mathbb{E}[C(\tau_1, \tau_2, \tau_3, H)^2(A_1(\tau_1, \tau_2, \tau_3, t, H)B_{\tau_1}^H + A_2(\tau_1, \tau_2, \tau_3, t, H)B_{\tau_2}^H  + A_3(\tau_1, \tau_2, \tau_3, t, H)B_{\tau_3}^H)
^2]dt\\
&= \frac{T^{2H+1}}{2H+1} - \frac{1}{4\tau_1^{2H}}\int_{\tau_1}^{\tau_2} (s^{2H}+\tau_1^{2H}-|s-\tau_1|^{2H})^2ds\\
&-C(\tau_1, \tau_2, H)^2\int_{\tau_2}^{\tau_3}[A_1(\tau_1, \tau_2, t, H)^2\tau_1^{2H} + A_2(\tau_1, \tau_2, t, H)^2\tau_2^{2H}\\
& - A_1(\tau_1, \tau_2, t, H)A_2(\tau_1, \tau_2, t, H) (\tau_2^{2H}+\tau_1^{2H}- |\tau_2-\tau_1|^{2H})]dt \\
& -\int_{\tau_3}^{T}\mathbb{E}[C(\tau_1, \tau_2, \tau_3, H)^2(A_1(\tau_1, \tau_2, \tau_3, t, H)B_{\tau_1}^H + A_2(\tau_1, \tau_2, \tau_3, t, H)B_{\tau_2}^H  + A_3(\tau_1, \tau_2, \tau_3, t, H)B_{\tau_3}^H)
^2]dt,
\end{align*}
where
\begin{align*}
&\int_{\tau_3}^{T}\mathbb{E}[ C(\tau_1, \tau_2, \tau_3, H)^2(A_1(\tau_1, \tau_2, \tau_3, t, H)B_{\tau_1}^H + A_2(\tau_1, \tau_2, \tau_3, t, H)B_{\tau_2}^H  + A_3(\tau_1, \tau_2, \tau_3, t, H)B_{\tau_3}^H)^2
]dt\\
=&C(\tau_1, \tau_2, \tau_3, H)^2 \int_{\tau_3}^{T} [A_1(\tau_1, \tau_2, \tau_3, t, H)^2\tau_1^{2H}+A_2(\tau_1, \tau_2, \tau_3, t, H)^2\tau_2^{2H} +A_3(\tau_1, \tau_2, \tau_3, t, H)^{2}\tau_3^{2H}\\
&+A_1(\tau_1, \tau_2, \tau_3, t, H)A_2(\tau_1, \tau_2, \tau_3, t, H)(\tau_1^{2H}+\tau_2^{2H}-(\tau_2-\tau_1)^{2H})\\
&+A_1(\tau_1, \tau_2, \tau_3, t, H)A_3(\tau_1, \tau_2, \tau_3, t, H)(\tau_1^{2H}+\tau_3^{2H}-(\tau_3-\tau_1)^{2H})\\
&+A_2(\tau_1, \tau_2, \tau_3, t, H)A_3(\tau_1, \tau_2, \tau_3, t, H)(\tau_2^{2H}+\tau_3^{2H}-(\tau_3-\tau_2)^{2H})
]dt.  
\end{align*}
\end{proof}

\subsection{Proof of Corollary \ref{prop-level2}}
\label{appendix:prop6}
\begin{proof}
For two samples case, based on equation \eqref{eq:modified-distortion}, the modified distortion is
\begin{align*}
J^*(\eta_1, \eta_2) = \frac{T^{2H+1}}{2H+1} -  \mathbb{E}\bigg[\int_{\tau_1}^{\tau_{2}} (\hat{B}_t^H)^2dt + \int_{\tau_2}^T (\hat{B}_t^H)^2dt \bigg],
\end{align*}
where
\begin{align*}
\int_{\tau_1}^{\tau_2}(\hat{B}_s^H)^2ds &=   \int_{\tau_1}^{\tau_2}(\frac{B_{\tau_1}^H}{2\tau_1^{2H}}(s^{2H}+\tau_1^{2H}-|s-\tau_1|^{2H}))^2ds\\
& = \frac{(B_{\tau_1}^H)^2}{4} \big[(\tau_2-\tau_{1}) + \frac{\tau_2^{4H+1}/\tau_1^{4H}-\tau_1}{4H+1} + \frac{2(\tau_2^{2H+1}-\tau_1^{2H+1})/\tau_1^{2H}}{2H+1} \\
& +\frac{(\tau_2 - \tau_1)^{4H+1}/\tau_1^{4H}}{4H+1} - 2\frac{(\tau_2 - \tau_1)^{2H+1}/\tau_1^{2H}}{2H+1} - 2\tau_1\int_1^{\frac{\tau_2}{\tau_1}}\mu^{2H}(\mu-1)^{2H}d\mu
\big],
\end{align*}
and
\begin{align*}
\int_{\tau_2}^{T}(\hat{B}_s^H)^2ds =&\int_{\tau_2}^T C(\tau_1, \tau_2, H)^2(A_1(\tau_1, \tau_2, t, H)B_{\tau_1}^H - A_2(\tau_1, \tau_2, t, H)B_{\tau_2}^H)^2ds\\
& = C(\tau_1, \tau_2, H)^2\int_{\tau_2}^T [A_1^2(\tau_1, \tau_2, t, H)(B_{\tau_1}^H)^2 +A_2^2(\tau_1, \tau_2, t, H)(B_{\tau_2}^H)^2 \\
&-2 A_1(\tau_1, \tau_2, t, H)A_2(\tau_1, \tau_2, t, H)B_{\tau_1}^HB_{\tau_2}^H]dt\\
& =   C(\tau_1, \tau_2, H)^2\big[(B_{\tau_1}^H)^2 \int_{\tau_2}^T A_1^2(\tau_1, \tau_2, t, H)dt + (B_{\tau_2}^H)^2 \int_{\tau_2}^T A_2^2(\tau_1, \tau_2, t, H)dt\\
& - 2B^H_{\tau_1}B^H_{\tau_2}\int_{\tau_2}^T A_1(\tau_1, \tau_2, t, H)A_2(\tau_1, \tau_2, t, H)dt\big]. 
\end{align*}
These can finish the proof.
\end{proof}

\subsection{Proof of Corollary \ref{prop-level3}}
\label{appendix:prop7}
\begin{proof}
For three samples case, based on equation \eqref{eq:modified-distortion}, the modified distortion is
\begin{align*}
J^*(\eta_1, \eta_2, \eta_3) &= \frac{T^{2H+1}}{2H+1} - \mathbb{E}\sum_{i = 1}^3 \int_{\tau_i}^{\tau_{i+1}} (\hat{B}_t^H)^2dt,
\end{align*}
where $\tau_4 = T$.
The calculation for $\int_{\tau_1}^{\tau_{2}} (\hat{B}_t^H)^2dt+\int_{\tau_2}^{\tau_{3}} (\hat{B}_t^H)^2dt$ is given in Section \ref{appendix:prop6} where we set $T = \tau_3$ in Section \ref{appendix:prop6}.

For $\tau_3 \le t < T$, the estimation is
\begin{align*}
\hat{B}_t^H &= C(\tau_1, \tau_2, \tau_3, H)(A_1(\tau_1, \tau_2, \tau_3, t, H)B_{\tau_1}^H + A_2(\tau_1, \tau_2, \tau_3, t, H)B_{\tau_2}^H  + A_3(\tau_1, \tau_2, \tau_3, t, H)B_{\tau_3}^H),
\end{align*}
where the calculation has been done in Section \ref{appendix:corollary2}.

Then we can obtain:
\begin{align*}
& \int_{\tau_3}^{T} (\hat{B}_t^H)^2dt\\
=& \int_{\tau_3}^{T}[ C(\tau_1, \tau_2, \tau_3, H)^2(A_1(\tau_1, \tau_2, \tau_3, t, H)B_{\tau_1}^H + A_2(\tau_1, \tau_2, \tau_3, t, H)B_{\tau_2}^H  + A_3(\tau_1, \tau_2, \tau_3, t, H)B_{\tau_3}^H)^2
]dt \\
=&C(\tau_1, \tau_2, \tau_3, H)^2 \big[ (B_{\tau_1}^H)^2 \int_{\tau_3}^{T} A_1(\tau_1, \tau_2, \tau_3, t, H)^2 dt + (B_{\tau_2}^H)^2 \int_{\tau_3}^{T} A_2(\tau_1, \tau_2, \tau_3, t, H)^2 dt \\
&+ (B_{\tau_3}^H)^2 \int_{\tau_3}^{T} A_3(\tau_1, \tau_2, \tau_3, t, H)^2 dt + 2B_{\tau_1}^HB_{\tau_2}^H  \int_{\tau_3}^{T} A_1(\tau_1, \tau_2, \tau_3, t, H)A_2(\tau_1, \tau_2, \tau_3, t, H) dt \\
&+ 2B_{\tau_1}^HB_{\tau_3}^H \int_{\tau_3}^{T} A_1(\tau_1, \tau_2, \tau_3, t, H)A_3(\tau_1, \tau_2, \tau_3, t, H) dt  \\
&+ 2B_{\tau_2}^HB_{\tau_3}^H \int_{\tau_3}^{T} A_2(\tau_1, \tau_2, \tau_3, t, H)A_3(\tau_1, \tau_2, \tau_3, t, H) dt \big],
\end{align*}
which can finish the proof.
\end{proof}

\bibliographystyle{ECA_jasa}
\bibliography{cite}
\end{document}